\documentclass[11pt]{article}%
\usepackage{amssymb}
\usepackage{amsfonts}
\usepackage{amsmath}
\usepackage{afterpage}
\usepackage{latexsym}
\usepackage{afterpage}
\usepackage{graphicx}
\usepackage{epstopdf}
\usepackage{hyperref}
\usepackage{float}
\usepackage{appendix}
\restylefloat{table}

\usepackage[nohead]{geometry}

\usepackage[doublespacing]{setspace}

\usepackage[bottom]{footmisc}
\usepackage{indentfirst}
\usepackage{endnotes}
\usepackage{booktabs}
\usepackage{rotating}
\usepackage{comment}
\newtheorem{theorem}{Theorem}

\newtheorem{corollary}{Corollary}

\newtheorem{definition}{Definition}
\newtheorem{example}{Example}

\newtheorem{lemma}{Lemma}

\newtheorem{remark}{Remark}

\newenvironment{proof}[1][Proof]{\noindent\textbf{#1.} }{\ \rule{0.5em}{0.5em}}

\makeatletter
\def\@biblabel#1{\hspace*{-\labelsep}}
\makeatother
\usepackage{natbib} 
\hypersetup{
colorlinks=true,%
citecolor=red,%
filecolor=black,%
linkcolor=red,%
urlcolor=black
}

\begin{document}

\title{Partial Identification of Distributional Parameters\\in Triangular Systems}
\author{Ju Hyun Kim\thanks{Department of Economics, University of North Carolina at Chapel Hill, Chapel Hill, NC, 27599. Email: juhkim@email.unc.edu. I am greatly
indebted to my advisor Bernard Salani\'{e} for his guidance and support
throughout this project. This paper has benefited from discussions with Andrew
Chesher, Toru Kitagawa, Dennis Kristensen, and Ismael Mourifi\'{e}. I thank participants of the econometrics colloquiums at Columbia. All errors are mine.}}
\date{\today }
\maketitle

\begin{abstract}
I study partial identification of distributional parameters in triangular
systems. The model consists of a nonparametric outcome equation and a
selection equation. This allows for general unobserved heterogeneity in
potential outcomes and selection on unobservables. The distributional
parameters considered in this paper are the marginal distributions of
potential outcomes, their joint distribution, and the distribution of
treatment effects. I explore different types of plausible restrictions to
tighten existing bounds on these parameters. The restrictions include stochastic
dominance, quadrant dependence between unobservables, and monotonicity between
potential outcomes. My identification applies to the whole population without
a full support condition on instrumental variables and does not rely on rank
similarity. I also provide numerical examples to illustrate identifying
power of the restrictions. \newline\textbf{Keywords:} Partial Identification,
Triangular Systems, Stochastic Monotonicity, Monotone Treatment Response,
Quadrant Dependence \newline\textbf{JEL Classifications:} C14, C21, C61, C81.

\end{abstract}

\section{Introduction}
\label{c2s1}
In this paper, I consider partial identification of distributional
parameters in triangular systems as follows: 
\begin{eqnarray*}
Y &=&m\left( D,\varepsilon _{D}\right) \text{,} \\
D &=&\boldsymbol{1}\left[ p\left( Z\right) \geq U\right] .
\end{eqnarray*}%
Here $Y$ denotes a continuous observed outcome, $D$ a binary selection
indicator, $Z$ instrumental variables (IV), $\varepsilon _{D}$ a
scalar unobservable, and $U$ a scalar unobservable. Let $Y_{0}$ and $Y_{1}$
denote the potential outcomes without and with some treatment, respectively,
with $Y_{d}=m\left( d,\varepsilon _{d}\right) $ for $d\in \left\{
0,1\right\}$. Note that I suppress covariates included in the outcome equation and the selection equation to keep the notation manageable. The analysis readily extends to accoount for conditioning on these covariates. The distributional parameters that I am interested in are the
marginal distributions of $Y_{0}$ and $Y_{1}$, their joint distribution, and
the distribution of treatment effects (DTE) $P\left( \Delta \leq \delta \right) $ with the treatment effect $%
\Delta =Y_{1}-Y_{0}$ and $\delta \in \mathbb{R}$.

In the context of welfare policy evaluation, various distributional
parameters beyond the average effects are often of fundamental interest.
First, changes in marginal distributions of potential outcomes induced by
policy are one of the main concerns when the impact on total social welfare
is calculated by comparing the distributions of potential outcomes. Examples
include inequality measures such as the Gini coefficient and the Lorenz
curve with and without policy (e.g. \citet{B2007}), and stochastic
dominance tests between the distributions of potential outcomes (e.g. \citet{A2002}). Second, information on the joint distribution
of $Y_{0}$ and $Y_{1}$, and the DTE beyond their marginal distributions is
often required to capture individual specific heterogeneity in program
evaluation. Examples of such information include the distribution of the outcome with treatment given that the potential outcome without
treatment lies in a specific set $P\left( Y_{1}\leq y_{1}|Y_{0}\in \Upsilon
_{0}\right) $ for some set $\Upsilon _{0}$ in $\mathbb{R}$, the fraction of
the population that benefits from the program $P\left( Y_{1}\geq
Y_{0}\right) ,$ the fraction of the population that has gains or losses in a
specific range $P\left( \delta ^{L}\leq Y_{1}-Y_{0}\leq \delta ^{U}\right) $
for $\left( \delta ^{L},\delta ^{U}\right) \in $ $\mathbb{R}^{2}$ with $%
\delta ^{L}\leq \delta ^{U}$, and the $q$ quantile of the impact
distribution $\inf \left\{ \delta :F_{\Delta }\left( \delta \right)
>q\right\} $.

The triangular system considered in this study consists of an
outcome equation and a selection equation. This structure allows for general
unobserved heterogeneity in potential outcomes and selection on
unobservables. The error term in the outcome equation represents unobserved
factors causing heterogeneity in potential outcomes among observationally
equivalent individuals.\footnote{%
Since it determines the relative ranking of such individuals in the
distribution of potential outcomes, it is also referred to as the rank
variable in the literature. See \citet{CH2013}.} The
selection model with a latent index crossing a threshold has been widely
used to model selection into programs. In the model, the latent index $%
p\left( Z\right) -U$ is interpreted as the net expected utility from
participating in the program. \citet{V2002} showed that the model is equivalent to the local average
treatment effect (LATE) framework developed by \citet{IA1994}.\footnote{%
The LATE\ framework consists of two main assumptions: independence and
monotonicity. The former assumes that the instrument is jointly independent
of potential outcomes and potential selection at each value of the
instrument, while the latter assumes that the instrument affects the
selection decision in the same direction for every individual. Since the contribution of \citet{V2002}, the selection structure has been widely recognized
as the model which is not only motivated by economic theory but also as weak
as LATE assumptions.}

In the literature, the identification method has relied on either the full support of IV or rank similarity to
consider the $entire$ population. The full support condition requires IV to change the probability of receiving the treatment from zero to one.\footnote{This type of identification is also referred to as identification at infinity.} As discussed in \citet{H1990}, and \citet{IW2009},  however, the applicability of the identification results is very limited because such instruments are difficult to find in practice. Rank similarity assumes that the
distribution $\varepsilon _{d}$ conditional on $U$ does not depend on $d$ for $d\in \left\{
0,1\right\}$. As
a relaxed version of rank invariance, it allows for a random variation
between ranks with and without treatment.\footnote{%
In this sense, rank similarity is also called \emph{expectational} rank
invariance. See \citet{CH2013}. \citet{BSV2012}, \citet{BSV2008}, \citet{SV2011}, and \citet{M2013} made use of rank
similarity to identify average treatment effects for models with a binary
outcome variable. Note that these results are readily extended to
identification of marginal distributions for continuous outcome variables.}
However, rank similarity is invalid when individuals select
treatment status based on their potential outcomes, as in the Roy
model.

The literature on identification in triangular systems has stressed marginal distributions more than the joint distribution or the
DTE. \citet{H1990} point-identified marginal distributions relying on the full support condition. \citet{CH2005} showed that the marginal distributions are point-identified for the entire population under rank similarity. Without these conditions, most of the
literature has focused on \emph{local} identification for \emph{compliers},
to circumvent complications in considering the whole population. \citet{IR1997}, and \citet{A2002} showed that under LATE assumptions presented by \citet{IA1994},
marginal distributions of potential outcomes are point-identified for \emph{%
compliers} who change their selection in a certain direction according to
the change in the value of IV.  \citet{K2009} contrasts with other work in the
sense that his identification is for the $entire$ population without relying
on the full support of IV and rank similarity. He obtained the
identification region for the marginal distributions under IV conditions.%
\footnote{%
The IV restrictions that he considers are (i) IV independence of each
potential outcome, (ii) IV joint independence of the pair of potential
outcomes, and (iii) LATE restrictions.} The joint distribution and the DTE
have not been investigated in these studies.

The literature on identification of the joint distribution and the DTE\ is
relatively small. \citet{FW2010} established sharp bounds on the joint
distribution and the DTE in semiparametric triangular systems using Fr\'{e}%
chet-Hoeffding bounds and Makarov bounds, respectively. Their identification
is for the entire population under the full support of IV. Also, \citet{GH2012} point-identified the DTE based on a random coefficients
specification for the selection equation. To do this, they also relied on
the full support of the IV. \citet{P2013} studied identification of the joint
distribution and the DTE in the extended Roy model, a particular case of
triangular systems.\footnote{%
The extended Roy model models individual self-selection based on the
potential outcomes and observable characteristics without allowing for any
additional selection unobservables.} Although he point-identified the joint
distribution and the DTE by taking advantage of the particular structure of
the extended Roy model, his identification only applies to the group of
compliers. \citet{HSC1997}, \citet{CHH2003}, and Aakvik et al.
(2005) considered factor structures in outcome unobservables and assumed the
presence of additional proxy variables to identify the joint distribution.
\citet{HM2014} considered Roy models with a binary outcome
variable. They derived sharp bounds on the marginal distributions and the
joint distribution of the potential outcomes. Although they did not assume
the full support of IV and rank similarity, for the joint distribution
bounds they focused on a one-factor structure, as proposed in
\citet{AHV2005}.

    The main contribution of this paper is to partially identify the joint distribution and the DTE as well as marginal distributions for the entire population without the full support condition of IV and rank similarity. To avoid strong assumptions and impose plausible information on the model, I consider weak restrictions on dependence between unobservables and between potential outcomes. First, I obtain sharp bounds on the distributional parameters for the worst case, which only assumes the latent index model of \citet{V2002}. Next, I explore three different types of restrictions to tighten the worst bounds and investigate how each restriction contributes to improving the identification regions of these parameters.

The first restriction that I consider is negative stochastic monotonicity (NSM) between $\varepsilon _{d}$ and $U$ for $d\in \left\{0,1\right\}$. NSM means that $\varepsilon _{d}$  increases as U increases for $d\in \left\{0,1\right\}$. This assumption has been adopted in the literature including \citet{JPX2011} for its plausibility in practice.\footnote{\citet{C2005} also considered stochastic monotonicity to identify triangular systems with a multivalued discrete endogenous variable. However, his setting does not allow for the binary selection.} The role of NSM in my paper is different from theirs: I use this condition to bound the counterfactual marginal distributions for the whole population, while they use this condition to identify a particular structure in the outcome equation for individuals who change their selection by variation in IV. Another type of restriction that I discuss is conditional positive quadrant dependence (CPQD) for the dependence between  $\varepsilon _{0}$ and $\varepsilon _{1}$ conditional on $U$. CPQD means that $\varepsilon _{0}$ and $\varepsilon _{1}$ are positively dependent conditionally on $U$. Finally, I consider monotone treatment response (MTR) $P\left( Y_{1}\geq Y_{0}\right) =1$, which assumes that each individual benefits from the treatment. Unlike other two restrictions, MTR restricts the support of potential outcomes. 
 
Interesting conclusions emerge from the results of this paper. First, NSM has identifying power on the marginal distributions only. CPQD improves the bounds on the joint distribution only. On the other hand, MTR yields substantially tighter identification regions for all three distributional parameters.

In the next section, I\ give a formal description of my problem, define the
parameters of interest, and discuss assumptions considered for the
identification. In Section \ref{c2s3}, I\ establish sharp bounds on the
distributional parameters. Section \ref{c2s4} discusses testable implications and
considers bounds when some of the restrictions are jointly imposed. Section
\ref{c2s5} provides numerical examples to illustrate the identifying power of each
restriction and Section \ref{c2s6} concludes. Technical proofs are collected in Appendix.

\section{Basic Model and Assumptions}
\label{c2s2}
\subsection{Model}
\label{c2s2s1}
Consider the triangular system:%
\begin{eqnarray}
Y &=&m\left( D,\varepsilon _{D}\right) ,  \label{model} \\
D &=&\boldsymbol{1}\left[ p\left( Z\right) \geq U\right] ,  \notag
\end{eqnarray}%
where $Y$ is an observed scalar outcome, $D$ is a binary indicator for
treatment participation, $\varepsilon _{D}$ is a scalar unobservable in the
outcome equation$,$ and $U$ is a scalar unobservable in a selection
equation. Since $Y$ is an realized outcome as a result of selection $D,$ $Y$
can be written as $Y=D\times Y_{1}+(1-D)\times Y_{0}$, where $Y_{0}$ and $%
Y_{1}$ are potential outcomes for the treatment status $0$ and $1$,
respectively.\ Let $Z$ denote a scalar or vector-valued IV that is excluded
from the outcome equation and $\mathcal{Z}$ denote the support of $Z$.$\ $For each 
$z\in \mathcal{Z} $, let $D_{z}$\ be the potential treatment participation when $Z=z$%
.\ 

Note that I allow the distribution of outcome unobservables to vary with the
selection $D$. Also, I\ do not impose an additively separable structure on
the unobservable in the outcome equation. In the selection equation, $%
p\left( Z\right) -U$ can be interpreted as the net utility from treatment
participation.\footnote{%
\citet{V2006} showed that selection equation in the model (\ref{model}) is
equivalent to the most general form of the latent index selection model \ $D=%
\boldsymbol{1}\left[ s\left( Z,V\right) \geq 0\right] $ where $s$ is unknown
function and $V$ is a (possibly) vector-valued unobservable under
monotonicity of the selection in the instruments. Technically, the condition
means that for any $z$ and $z^{\prime }$ in $\mathcal{Z}$, if $s\left(
z,v_{0}\right) >s\left( z^{\prime },v_{0}\right) $ for some $v_{0}\in 
\mathcal{V},$ $s\left( z,v\right) >s\left( z^{\prime },v\right) $ for almost
every value of $v\in \mathcal{V}$ where $\mathcal{V}$ is the support of $V.$
Intuitively, this implies that the sign of the change in net utility caused
by the instruments does not depend on the value of the unobservable $V$.}
Note that selection on unobservables arises from dependence between $%
\varepsilon _{D}$ and $U.$

\begin{remark}
Without loss of generality, I assume that $U\sim Unif\left( 0,1\right) $ for
normalization. Then $p\left( z\right) =P\left[ D=1|Z=z\right] $ is
interpreted as a propensity score.
\end{remark}

Throughout this study, I\ impose the following assumptions on the model (\ref%
{model}).

\begin{description}
\item[M.1 (Monotonicity)] $m\left( d,\varepsilon _{d}\right) $ is strictly
increasing in a scalar unobservable $\varepsilon _{d}$ for each $d\in
\left\{ 0,1\right\} .$

\item[M.2 (Continuity)] For $d\in \left\{ 0,1\right\} $, the distribution function of $\varepsilon _{d}$ is
absolutely continuous with respect to the Lebesgue measure on $\mathbb{R}$.

\item[M.3 (Exogeneity)] $Z\perp \!\!\!\perp \left( \varepsilon
_{0},\varepsilon _{1},U\right) $.

\item[M.4 (Propensity Score)] The function $p\left( \cdot \right) $ is a
nonconstant and continuous function for the continuous element in $Z$.
\end{description}

$M.1$ and $M.2$ ensure the continuous distribution of $Y_{d}$ and
invertibility of the function $m\left( d,\varepsilon _{d}\right) $ in the
second argument, which is a standard assumption in the literature on
nonparametric models with a nonseparable error. $M.3$ is an instrument
exogeneity condition. That is, the instrument $Z$ exogenously affects
treatment selection and it affects the outcome only through the treatment
status. Furthermore, $Z$ does not affect dependence among unobservables $%
\varepsilon _{0},\varepsilon _{1}$, and $U$. $M.4$ is necessary to ensures sharpness of the bounds. It requires that when some elements of the IV are continuous, the propensity score function $p\left( \cdot \right) $ be continuous for the continuous elements of IV when the discrete elements of IV are held constant. See \citet{SV2011} for details.
\begin{remark}
\citet{V2002} showed that under $M.3$, the selection equation in the model (%
\ref{model}) is equivalent to the assumptions in the LATE framework
developed by \citet{IA1994}:\ independence and monotonicity. The
LATE independence condition assumes that $Z\perp \!\!\!\perp \left(
Y_{0},Y_{1},U\right) $ and that the propensity score $p\left( z\right) $ is
a nonconstant function. The LATE monotonicity condition assumes\ that either $%
D_{z}\geq D_{z^{\prime }}$ or $D_{z^{\prime }}\geq D_{z}$ with probability
one for $\left( z,z^{\prime }\right) \in \mathcal{Z} \times \mathcal{Z} $ with $z\neq
z^{\prime }.$
\end{remark}

Numerous examples fit into the model (\ref{model}).\ I\ refer to the
following three examples throughout the paper.

\begin{example}
(The effect of job training programs on wages) Let $Y$ be a wage and $D$ be
an indicator of enrollment for the program. Let $Z$ be the random assignment
for the training service when the program designs randomized offers in the
early application process. Note that such a randomized assignment has been
widely used as a valid instrument in the LATE framework, which is equivalent
to the model (\ref{model}) considered in this paper.
\end{example}

\begin{example}
(College premium) Let $Y$ be a wage and $D$ be the college education
indicator. The literature including \citet{CHV2011} has
used the distance to college, local wage, local unemployment rate,\ and
average tuition for public colleges in the county of residence as IV.
\end{example}

\begin{example}
(The effect of smoking on infant birth weight) Let $Y$ be an infant birth
weight and $D$ be a smoking indicator. In the empirical literature, state
cigarette taxes, policy interventions including tax hikes, and randomized
counselling have been used as IV.
\end{example}

\subsection{Objects of Interest and Assumptions}
\label{c2s2s2}
The objects of interest here are the marginal distribution functions of $%
Y_{0}$ and $Y_{1},$ $F_{0}\left( y_{0}\right) $ and $F_{1}\left(
y_{1}\right) $, their joint distribution function $F\left(
y_{0},y_{1}\right) $, and the DTE $F_{\Delta }\left( \delta \right) =P\left(
Y_{1}-Y_{0}\leq \delta \right) $ for fixed $y_{0}$, $y_{1}$, and $\delta $
in $\mathbb{R}$. I\ obtain sharp bounds on $F_{0}\left( y_{0}\right) ,$ $%
F_{1}\left( y_{1}\right) ,$ $F\left( y_{0},y_{1}\right) ,$ and $F_{\Delta
}\left( \delta \right) $ under various weak restrictions. First, I\ derive
worst case bounds making use of only $M.1-M.4$ in the model (\ref{model}).
The conditions $M.1-M.4$ are maintained throughout this study. Second, I\
impose negative stochastic monotonicity (NSM) between each outcome unobservable and
the selection unobservable, and show how identification regions improve
under the additional restriction. Third, I\ consider conditional positive
quadrant dependence (CPQD) as a restriction between two outcome
unobservables $\varepsilon _{0}$ and $\varepsilon _{1}$ conditional on the
selection unobservable $U$. I\ also explore identifying
power of this restriction on each parameter, when it is imposed on top of $%
M.1-M.4$. Lastly, I\ consider monotonicity between two potential outcomes
as a different type of restriction. Henceforth, I\ call this monotone
treatment response (MTR). I derive sharp bounds under MTR in addition to $%
M.1-M.4$.

First, I\ present the definition of NSM, CPQD, and MTR. I\ also illustrate them using a toy model and discuss the underlying
intuition with economic examples.

\begin{description}
\item[NSM (Negative Stochastic Monotonicity)] Both $\varepsilon _{0}$ and $\varepsilon
_{1}$ are first order stochastically nonincreasing in $U.$ That is, $P\left(
\varepsilon _{d}\leq e|U=u\right) $ is nondecreasing in $u\in \left(
0,1\right) $ for any $e\in \mathbb{R}$ and $d\in \left\{ 0,1\right\} .$

\item[CPQD (Conditional Positive Quadrant Dependence)] $\varepsilon _{0}$ and 
$\varepsilon _{1}$ are positively quadrant dependent conditionally on $U.$
That is, for $\left( \varepsilon _{0},\varepsilon _{1}\right) \in \mathbb{%
R\times R}$ and $u\in \left( 0,1\right) ,$ 
\begin{equation*}
P\left[ \varepsilon _{0}\leq e_{0},\varepsilon _{1}\leq e_{1}|U=u\right]
\geq P\left[ \varepsilon _{0}\leq e_{0}|U=u\right] P\left[ \varepsilon
_{1}\leq e_{1}|U=u\right] .
\end{equation*}
\end{description}
To better understand these restrictions, consider a particular case
where $\varepsilon _{0}$ and $\varepsilon _{1}$ have a one-factor
structure as follows: for $d\in \left\{ 0,1\right\} $%
\begin{equation}
\varepsilon _{d}=\rho _{d}U+\nu _{d},  \label{factor}
\end{equation}%
where $\left( \nu _{0},\nu _{1}\right) \perp \!\!\!\perp U.$ Here$\ U$ is the
unobservable in the selection equation, while $\nu _{0}$ and $\nu _{1}$ represent treatment
specific heterogeneity.\footnote{%
This one-factor structure has been discussed in the context of the
effects of employment programs in the literature including \citet{AHV2005} and \citet{HM2014}.}

In this setting, NSM requires that $\rho _{0}$ and $\rho _{1}$
be nonpositive. Note that the direction of the sign of the monotonicity is not crucial because
my identification strategy can be applied to negative stochastic monotonicity.
Intuitively, NSM implies that as the level of $U$\ increases, both $%
\varepsilon _{0}$ and $\varepsilon _{1}$ decrease or stay constant. This
condition is plausible in many empirical applications. In job
training programs, individuals with higher motivation for the training
program (lower $U$) are more likely to invest effort in
their work (higher $\varepsilon _{0}$ and $\varepsilon _{1}$) than others
with lower motivation (higher $U$). In the example of the college premium,
a lower reservation utility (lower $U$) for college
education $\left( D=1\right) $ is more likely to go with a higher level of
unobserved abilities $($higher $\varepsilon _{0}$ and $\varepsilon _{1})$.
Regarding the effect of smoking on infant birth weight, NSM suggests that
controlling for observed characteristics, individuals with a lower desire
(lower $U$) for smoking ($D=0$) are more likely to have a healthier
lifestyle $($higher $\varepsilon _{1}$ and $\varepsilon _{0})$ than those
with a higher desire (higher $U$).

CPQD excludes any negative dependence between $\nu _{0}$ and $\nu _{1}$ in
the example (\ref{factor}). Before discussing implications of CPQD, I
present the concept of \emph{quadrant dependence}. Quadrant dependence
between two random variables is defined as follows:

\begin{definition}
(Positive (Negative) Quadrant Dependence, \citet{L1966}) Let $X$ and $Y$ be
random variables. $X$ and $Y$ are positively (negatively) quadrant dependent
if for any $(x,y)\in \mathbb{R}^{2}$, 
\begin{equation*}
P\left[ X\leq x,Y\leq x\right] \geq \left( \leq \right) P\left[ X\leq x%
\right] P\left[ Y\leq x\right] .
\end{equation*}%
or equivalently,%
\begin{equation*}
P\left[ X>x,Y>x\right] \geq \left( \leq \right) P\left[ X>x\right] P\left[
Y>x\right] .
\end{equation*}
\end{definition}

Intuitively, $X$ and $Y$ are positively quadrant dependent, if
the probability that they are simultaneously small or large is at least as high as it would be if they were independent.\footnote{%
For details, see pp. 187-188 in \citet{N2006}.} Note that quadrant
dependence is a very weak dependence measure among a variety of dependence
concepts in copula theory.\footnote{%
NSM is a stronger concept of dependence between two random variables than quadrant dependence. If $X$ and $%
Y$ are first order stochastically nondecreasing in $Y$ and $X$,
respectively, then $X$ and $Y$ are positively quadrant dependent.}

I\ impose \textit{conditional } positive quadrant dependence between $\varepsilon _{0}$ and $\varepsilon _{1}$ given the selection
unobservable $U$. In the example (\ref{factor}), CPQD requires that $\nu_{0}$ and $\nu_{1}$ be positively quadrant dependent.
Note that CPQD is satisfied even when $\nu _{0}$ and $\nu_{1}$\ are
independent of each other.

To intuitively understand the implications of CPQD, consider the example (%
\ref{factor}) for the three examples. For the example of job training
programs, suppose that two agents A and B have the same level of motivation for
the program and the identical observed characteristics. CPQD implies that if
the agent A is likely to earn more than agent B when they both
participate in the program, then A is still likely to earn more than B if
neither A nor B participates. This is due to the nonnegative correlation
between $\nu _{0}$\ and $\nu _{1}.$ In the college premium example, the
selection unobservable $U$ and another unobservable factor $\nu _{d}$ for $%
d\in \left\{ 0,1\right\} $ have been interpreted as an unobserved talent and
market uncertainty, respectively, in the literature including Jun et al.
(2012). CPQD excludes the case where market uncertainty unobservables $\nu
_{0}$\ and $\nu _{1}$ are negatively correlated. In the context of the
effect of smoking, after controlling for the desire for smoking and all
observed characteristics, the smoking (non-smoking) mother whose infant has
higher birth weight is more likely to have a heavier infant if she were a
non-smoker (smoker). Infant's weight is affected by mother's genetic factors 
$\nu _{d}$ for $d\in \left\{ 0,1\right\} ,$\ which are independent of her
preference for smoking. CPQD requires that mother's genetic factors in
treatment status $0$ and $1,$ $\nu _{0}$ and $\nu _{1}$ are nonnegatively
correlated.

\begin{description}
\item[MTR (Monotone Treatment Response)] $P\left( Y_{1}\geq Y_{0}\right) =1.$
\end{description}

MTR indicates that every individual benefits from some program or treatment.
MTR has been widely adopted in empirical research on evaluation of welfare
policy and various treatments including three examples I\ consider, the
effect of funds for low-ability pupils (\citet{H2012}), the impact of the
National School Lunch Program on child health (\citet{GKP2011}), and
various medical treatments (\citet{BSV2008}, \citet{BSV2012}). \ \ \ \ \ \ \
\ \ \ 

\subsection{Classical Bounds}
\label{c2s2s3}
In this subsection, I present two classical bounds that are applicable
to bounds on the joint distribution function and bounds on the DTE when the
marginal distributions of $Y_{0}$ and $Y_{1}$ are given. These are referred
to frequently throughout the paper.

Suppose that marginal distributions $F_{0}$ and $F_{1}$ are given and no
other restriction is imposed on the joint distribution $F$. Sharp bounds on
the joint distribution $F$ are given as follows: for $\left(
y_{0},y_{1}\right) \in \mathbb{R}\times \mathbb{R},$ 
\begin{equation*}
\max \left\{ F_{0}\left( y_{0}\right) +F_{1}\left( y_{1}\right) -1,0\right\}
\leq F\left( y_{0},y_{1}\right) \leq \min \left\{ F_{0}\left( y_{0}\right)
,F_{1}\left( y_{1}\right) \right\} .
\end{equation*}%
These bounds are referred to as Fr\'{e}chet-Hoeffding bounds. The lower
bound is achieved when $Y_{0}$ and $Y_{1}$ are perfectly negatively
dependent, while the upper bound is achieved when they are perfectly
positively dependent.\footnote{$Y_{0}$ and $Y_{1}$ are perfectly positively
dependent if and only if $F_{0}(Y_{0})=F_{1}(Y_{1})$ with probability one,
and they are perfectly negatively dependent if and only if $%
F_{0}(Y_{0})=1-F_{1}(Y_{1})$ with probability one.}

Next, let 
\begin{eqnarray*}
F_{\Delta }^{L}\left( \delta \right) &=&\sup_{y}\max \left( F_{1}\left(
y\right) -F_{0}\left( y-\delta \right) ,0\right) , \\
F_{\Delta }^{U}\left( \delta \right) &=&1+\inf_{y}\min \left( F_{1}\left(
y\right) -F_{0}\left( y-\delta \right) ,0\right) .
\end{eqnarray*}%
Then for the DTE $F_{\Delta }\left( \delta \right) =P\left( \Delta \leq
\delta \right) =P\left( Y_{1}-Y_{0}\leq \delta \right) ,$ 
\begin{equation*}
F_{\Delta }^{L}\left( \delta \right) \leq F_{\Delta }\left( \delta \right)
\leq F_{\Delta }^{U}\left( \delta \right) ,
\end{equation*}%
and both $F_{\Delta }^{L}\left( \delta \right) $ and $F_{\Delta }^{U}\left(
\delta \right) $ are sharp. These bounds are referred to as Makarov bounds.

\section{Sharp Bounds}
\label{c2s3}
This section establishes sharp bounds on the marginal distributions of $Y_{0}$
and $Y_{1}$, the joint distribution and the DTE. I\ start with the worst case
bounds which are established under $M.1-M.4$ for model (\ref{model}).\ I
then\ obtain bounds under NSM and $M.1-M.4$, bounds under CPQD and $M.1-M.4,$
and finally those under MTR in addition to $M.1-M.4$. To compress long
notation, henceforth I\ refer to $P\left(  Y\leq y|D=d,Z=z\right)  $,
$P\left(  Y_{d}\leq y|D=1-d,Z=z\right)  $, $P\left(  Y\leq y,D=d|Z=z\right)
,$ \ and $P\left(  Y_{d}\leq y,D=1-d|Z=z\right)  $ as $P\left(  y|d,z\right)
$, $P_{d}\left(  y|1-d,z\right)  ,$ $P\left(  y,d|z\right)  $, and
$P_{d}\left(  y,1-d|z\right)  ,$ respectively, for $d\in\left\{  0,1\right\}
$, $y\in\mathbb{R}$,\ and $z\in\mathcal{Z}$.

\subsection{Worst Case Bounds}
\label{c2s3s1}
\citet{BGIM2007} obtained sharp bounds on marginal distributions of
$Y_{0}$ and $Y_{1}$ under $M.1-M.4.$ I\ take their approach to bounding the
marginal distributions. Given $M.3$, marginal distributions of $Y_{0}$ and
$Y_{1}$ can be written as follows: for each $z\in\mathcal{Z}$ and any
$y\in\mathbb{R}$,
\begin{align}
F_{1}\left(  y\right)   &  =P\left(  Y_{1}\leq y|Z=z\right)  \label{1}\\
&  =P\left(y,1|z\right)  +P_{1}\left(y,0|z\right)  .\nonumber
\end{align}
While the probability $P\left(y,1|z\right)  $ is observed, the
counterfactual probability $P_{1}\left(y,0|z\right)  $ is never
observed. Let $\overline{p}=\underset{z\in\mathcal{Z}}{\sup}p\left(  z\right)  ,$
$\underline{p}=\underset{z\in\mathcal{Z}}{\inf}p\left(  z\right)$. Note that
$\overline{p}$ and $\underline{p}$ are well defined under $M.4$.

For $z\in\mathcal{Z}$ such that $p\left(  z\right)  <\overline{p},$ the
counterfactual probability $P_{1}\left(y,0|z\right)  $ can be
decomposed as follows:%
\begin{align}
&  P_{1}\left( y,0|z\right)  \label{1.1}\\
&  =P\left(  Y_{1}\leq y,p\left(  z\right)  <U|z\right)  \nonumber\\
&  =P\left(  Y_{1}\leq y,p\left(  z\right)  <U\right)  \nonumber\\
&  =P\left(  Y_{1}\leq y,p\left(  z\right)  <U\leq\overline{p}\right)
+P\left(  Y_{1}\leq y,\overline{p}<U\right)  ,\nonumber
\end{align}
The second equality follows from $M.3$.

Note that $P\left(  Y_{1}\leq y,p\left(  z\right)  <U\leq\overline{p}\right)
$ is point-identified as follows:%
\begin{align*}
P\left(  Y_{1}\leq y,p\left(  z\right)  <U\leq\overline{p}\right)   &
=P\left(  Y_{1}\leq y,U\leq\overline{p}\right)  -P\left(  Y_{1}\leq y,U\leq
p\left(  z\right)  \right)  \\
&  =\underset{p\left(  z\right)  \rightarrow\overline{p}}{\lim}P\left(
y|1,z\right)  \overline{p}-P\left(  y|1,z\right)  p\left(  z\right)  .
\end{align*}
However, $P\left(  Y_{1}\leq y,\overline{p}<U\right)  $ is never observed.
Note that for
\[
P\left(  Y_{1}\leq y,\overline{p}<U\right)  =\underset{p\left(  z\right)
\rightarrow\overline{p}}{\lim}P_{1}\left(  y|0,z\right)  \left(
1-\overline{p}\right)  ,
\]
$\underset{p\left(  z\right) \rightarrow\overline{p}}{\lim}P_{1}\left(  y|0,z\right)$ can be any value between $0$ and $1$.
Therefore, I\ can derive bounds on $P\left(  Y_{1}\leq y,\overline
{p}<U\right)  $\ by\ plugging $0$ and $1$ into the counterfactual distribution
$P\left(  y|0,\overline{z}\right)  $. Similarly, the other counterfactual
probability $P_{0}\left(  y,1|z\right)  $ can be partially identified.

\begin{lemma}
[\citet{BGIM2007}]\label{L1.5}Under $M.1-M.4$, for any $z\in\mathcal{Z}$,
$P_{0}\left(  y,1|z\right)  $ and $P_{1}\left(  y,0|z\right)  $ are bounded as
follows:%
\begin{align*}
P_{0}\left(  y,1|z\right)   &  \in\left[  L_{01}^{wst}\left(  y,z\right)
,U_{01}^{wst}\left(  y,z\right)  \right]  ,\\
P_{1}\left(  y,0|z\right)   &  \in\left[  L_{10}^{wst}\left(  y,z\right)
,U_{10}^{wst}\left(  y,z\right)  \right]  ,
\end{align*}
where
\begin{align*}
L_{01}^{wst}\left(  y,z\right)   &  =\underset{p\left(  z\right)
\rightarrow\underline{p}}{\lim}P\left(  y|0,z\right)  \left(  1-\underline
{p}\right)  -P\left(  y|0,z\right)  \left(  1-p\left(  z\right)  \right)  ,\\
U_{01}^{wst}\left(  y,z\right)   &  =\underset{p\left(  z\right)
\rightarrow\underline{p}}{\lim}P\left(  y|0,z\right)  \left(  1-\underline
{p}\right)  -P\left(  y|0,z\right)  \left(  1-p\left(  z\right)  \right)
+\underline{p},\\
L_{10}^{wst}\left(  y,z\right)   &  =\underset{p\left(  z\right)
\rightarrow\overline{p}}{\lim}P\left(  y|1,z\right)  \overline{p}-P\left(
y|1,z\right)  p\left(  z\right)  ,\\
U_{10}^{wst}\left(  y,z\right)   &  =\underset{p\left(  z\right)
\rightarrow\overline{p}}{\lim}P\left(  y|1,z\right)  \overline{p}-P\left(
y|1,z\right)  p\left(  z\right)  +1-\overline{p},
\end{align*}
and these bounds are sharp.
\end{lemma}

\begin{proof}
The proof is in Appendix.
\end{proof}

\begin{remark}
\label{pointid} If $\underline{p}=0,$ then $P_{0}\left(y,1|z\right)  $
is point-identified as $L_{01}^{wst}\left(  y,z\right)  =U_{01}^{wst}\left(
y,z\right)  .$ On the other hand, if $\overline{p}=1,$ then $P_{1}\left(
y,0|z\right)  $ is point-identified as $L_{10}^{wst}\left(  y,z\right)
=U_{10}^{wst}\left(  y,z\right)  .$ Therefore, when the instruments shift the
propensity score from $0$ to $1$, both counterfactual probabilities are
point-identified, and thus both marginal distributions of potential outcomes
are point-identified. This full support condition implies that treatment
participation is completely determined by instruments in the limits, and
unobservables do not exert any influence on treatment selection in the limits
of the propensity score. Therefore, the distributions of potential outcomes
are point-identified as they are point-identified in the absence of selection
on unobservables.
\end{remark}

Note that under $M.1-M.4$, the model (\ref{model}) does not impose any
restriction on dependence between $Y_{0}$ and $Y_{1}$. Hence,
Fr\'{e}chet-Hoeffding bounds and Makarov bounds can be employed to establish
sharp bounds on the joint distribution and the DTE, respectively.
Specifically, for any $z\in\mathcal{Z},$%
\begin{align}
&  F\left(  y_{0},y_{1}\right)  \label{jointdecompose}\\
&  =P\left(  Y_{0}\leq y_{0},Y_{1}\leq y_{1}|z\right)  \nonumber\\
&  =P\left(  Y_{0}\leq y_{0},Y_{1}\leq y_{1}|0,z\right)  \left(  1-p\left(
z\right)  \right)  +P\left(  Y_{0}\leq y_{0},Y_{1}\leq y_{1}|1,z\right)
p\left(  z\right)  .\nonumber
\end{align}
The first equality follows from $M.3$. Now Fr\'{e}chet-Hoeffding bounds can be
established on $P\left(  Y_{0}\leq y_{0},Y_{1}\leq y_{1}|0,z\right)  $ and
$P\left(  Y_{0}\leq y_{0},Y_{1}\leq y_{1}|1,z\right)  $ based on
point-identified $P\left(  y_{0}|0,z\right)  $\ and partially identified
$P_{1}\left(  y_{1}|0,z\right)  ,$\ and partially identified $P_{0}\left(
y_{0}|1,z\right)  $\ and point-identified $P\left(  y_{1}|1,z\right)  ,$ respectively.

Note that when marginal distributions are partially identified, sharp bounds
on the joint distribution are obtained by taking the union of
Fr\'{e}chet-Hoeffding bounds over all possible pairs of marginal
distributions. Similarly, the DTE can be written as
\begin{align*}
&  P\left(  Y_{1}-Y_{0}\leq\delta\right)  \\
&  =P\left(  Y_{1}-Y_{0}\leq\delta|z\right)  \\
&  =P\left(  Y_{1}-Y_{0}\leq\delta|0,z\right)  \left(  1-p\left(  z\right)
\right)  +P\left(  Y_{1}-Y_{0}\leq\delta|1,z\right)  p\left(  z\right)  ,
\end{align*}
and Makarov bounds can be applied to $P\left(  Y_{1}-Y_{0}\leq\delta
|0,z\right)  $ and $P\left(  Y_{1}-Y_{0}\leq\delta|1,z\right)  $ based on
point-identified $P\left(  y_{0}|0,z\right)  $\ and partially identified
$P_{1}\left(  y_{1}|0,z\right)  ,$\ and partially identified $P_{0}\left(
y_{0}|1,z\right)  $\ and point-identified $P\left(  y_{1}|1,z\right)  ,$ respectively.

The specific forms of sharp bounds on marginal distributions of $Y_{0}$ and
$Y_{1},$ their joint distribution, and the DTE  under $M.1-M.4$ are provided
in Theorem \ref{T1} in Appendix.

\subsection{Negative Stochastic Monotonicity}
\label{c2s3s2}
In this subsection, I\ additionally impose NSM on dependence between
$\varepsilon_{0}$ and $U$ and between $\varepsilon_{1}$ and $U$. I\ show that
NSM has additional identifying power for marginal distributions, but not on
the joint distribution nor on the DTE.

First, I\ use NSM to tighten the bounds on counterfactual probabilities
$P_{1}\left(  y,0|z\right)  $ and $P_{0}\left(  y,1|z\right)  $. Consider a
counterfactual distribution $P_{1}\left(  y|0,z\right) = P\left(  \varepsilon_{1}\leq m^{-1}\left(  1,y\right)  |p\left(  z\right)
<U\right)  $. If $p\left(  z\right)  <\overline{p},$ under NSM, for any
$\widehat{p}\left(  z\right)  \in(p\left(  z\right)  ,1],$
\[
P\left\{  \varepsilon_{1}\leq m^{-1}\left(  1,y\right)  |p\left(  z\right)
<U\right\}  \geq P\left\{  \varepsilon_{1}\leq m^{-1}\left(  1,y\right)
|p\left(  z\right)  <U\leq\widehat{p}\left(  z\right)  \right\}  .
\]

Since $P\left\{  \varepsilon_{1}\leq m^{-1}\left(  1,y\right)  |p\left(
z\right)  <U\leq\widehat{p}\left(  z\right)  \right\}  $ is nondecreasing in
$\widehat{p}\left(  z\right)  $ by NSM, for $z\in\mathcal{Z}$ $\setminus
p^{-1}\left(  \overline{p}\right)  ,$
the highest possible observable lower bound is obtained when $\widehat{p}\left(
z\right)  =\overline{p}$. Therefore by NSM, for any $z\in\mathcal{Z}\setminus
p^{-1}\left(  \overline{p}\right)  ,$ NSM implies%
\begin{align*}
&  P_{1}\left(  y|0,z\right)  \\
&  =P\left(  \varepsilon_{1}\leq m^{-1}\left(  1,y\right)  |p\left(  z\right)
<U\right)  \\
&  \geq P\left(  \varepsilon_{1}\leq m^{-1}\left(  1,y\right)  |p\left(
z\right)  <U\leq\overline{p}\right)  \\
&  =\frac{P\left(  \varepsilon_{1}\leq m^{-1}\left(  1,y\right)
,U\leq\overline{p}\right)  -P\left(  \varepsilon_{1}\leq m^{-1}\left(
1,y\right)  ,U\leq p\left(  z\right)  \right)  }{\overline{p}-p\left(
z\right)  }.
\end{align*}
Obviously, $P\left(  \varepsilon_{1}\leq m^{-1}\left(  1,y\right)
,U\leq\overline{p}\right)  $ and $P\left(  \varepsilon_{1}\leq m^{-1}\left(
1,y\right)  ,U\leq p\left(  z\right)  \right)  $ are point-identified as
$\underset{p\left(  z\right)  \rightarrow\overline{p}}{\lim
}P\left(  y,1|z\right) $ and $P\left(  y,1|z\right)  $ for any
$z\in\mathcal{Z}.$

Similarly, $P_{0}\left(  y|1,z\right)  =P\left(  \varepsilon_{0}\leq
m^{-1}\left(  0,y\right)  |U\leq p\left(  z\right)  \right)  $ and by NSM, for
any $z\in\mathcal{Z}\setminus p^{-1}\left(  \underline{p}\right)  $%
\begin{align*}
&  P\left(  \varepsilon_{0}\leq m^{-1}\left(  0,y\right)  |U\leq p\left(
z\right)  \right)  \\
&  \leq P\left(  \varepsilon_{0}\leq m^{-1}\left(  0,y\right)  |\underline
{p}<U\leq p\left(  z\right)  \right)  \\
&  =\frac{P\left(  \varepsilon_{0}\leq m^{-1}\left(  0,y\right)
,\underline{p}<U\right)  -P\left\{  \varepsilon_{0}\leq m^{-1}\left(
0,y\right)  ,p\left(  z\right)  <U\right\}  }{p\left(  z\right)
-\underline{p}}.
\end{align*}
Also, $P\left(  \varepsilon_{0}\leq m^{-1}\left(  0,y\right)  ,\underline
{p}<U\right)  $ and $P\left(  \varepsilon_{0}\leq m^{-1}\left(  0,y\right)
,p\left(  z\right)  <U\right)  $\ are point-identified as $\underset{p\left(  z\right)  \rightarrow\underline{p}}{\lim
}P\left(  y,0|z\right) $ and $P\left(  y,0|z\right)  $, respectively, for
any $z\in\mathcal{Z}.$\ These bounds are tighter than bounds obtained without NSM.

On the other hand, NSM has no additional identifying power on the upper
bound on $P_{1}\left(  y|0,z\right)  $ and the lower bound on $P_{0}\left(
y|1,z\right)  ,$ which means that these bounds under NSM are identical to those
obtained without NSM$.$

\begin{lemma}
\label{L2}Under $M.1-M.4$ and NSM, $P_{0}\left(  y,1|z\right)  $ and
$P_{1}\left(  y,0|z\right)  $ are bounded as follows:%
\begin{align*}
P_{0}\left(  y,1|z\right)   &  \in\left[  L_{01}^{wst}\left(  y,z\right)
,U_{01}^{sm}\left(  y,z\right)  \right]  ,\\
P_{1}\left(  y,0|z\right)   &  \in\left[  L_{10}^{wst}\left(  y,z\right)
,U_{10}^{sm}\left(  y,z\right)  \right]  ,
\end{align*}
where
\begin{align*}
L_{10}^{sm}\left(  y,z\right)   &  =\left\{
\begin{tabular}
[c]{ll}%
$\left(  \frac{\underset{p\left(  z\right)  \rightarrow\underline{p}}{\lim
}P\left(  y,1|z\right)  -P\left(  y,1|z\right)  }{\overline{p}-p\left(
z\right)  }\right)  \left(  1-p\left(  z\right)  \right)  ,$ & for any
$z\in\mathcal{Z}\setminus p^{-1}\left(  \overline{p}\right)  $,\\
$0,$ & for $z\in p^{-1}\left(  \overline{p}\right)  ,$%
\end{tabular}
\ \ \right.  ,\\
U_{01}^{sm}\left(  y,z\right)   &  =\left\{
\begin{tabular}
[c]{ll}%
$\left(  \frac{P\left(  y,0|z\right)  -\underset{p\left(  z\right)
\rightarrow\underline{p}}{\lim}P\left(  y,0|z\right)  }{p\left(  z\right)
-\underline{p}}\right)  p\left(  z\right)  ,$ & for any $z\in\mathcal{Z}%
\setminus p^{-1}\left(  \underline{p}\right)  ,$\\
$p\left(  z\right)  ,$ & for $z\in p^{-1}\left(  \underline{p}\right)  ,$%
\end{tabular}
\ \ \right.  ,
\end{align*}
and these bounds are sharp.
\end{lemma}

Now, sharp bounds on marginal distributions of $Y_{0}$ and $Y_{1}$ are
obtained by plugging the results in Lemma \ref{L2} into the counterfactual probabilities.

Note that under NSM, sharp bounds on the joint distribution and sharp bounds on
the DTE are still obtained from Fr\'{e}chet-Hoeffding bounds and Makarov bounds.
To illustrate this, consider the case where $\rho_{0}=\rho_{1}=0$ in the
example (\ref{factor}).\footnote{Note that NSM restricts the sign of
$\rho_{d\text{ }}$ as nonnegative for $d\in\left\{  0,1\right\}  .$} This case
satisfies NSM and NSM does not impose any restriction on the
dependence between $\nu_{0}$ and $\nu_{1}$. Therefore, sharp bounds on the
joint distribution and the DTE are obtained by the same token as in Subsection \ref{c2s3s1}.

The specific forms of sharp bounds on marginal distributions of $Y_{0}$ and
$Y_{1},$ their joint distribution, and the DTE under $M.1-M.4$ and NSM are
provided in Corollary \ref{T2} in Appendix.

\subsection{Conditional Positive Quadrant Dependence}
\label{c2s3s3}
Unlike NSM, CPQD has no additional identifying power for the joint
distribution and the DTE. In this subsection, I impose weak positive
dependence between $\varepsilon_{0}$ and $\varepsilon_{1}$ conditional on $U$
by considering CPQD as follows: for any $\left(  e_{0},e_{1}\right)
\in\mathbb{R}^{2},$
\begin{equation}
P\left[  \varepsilon_{0}\leq e_{0}|u\right]  P\left[  \varepsilon_{1}\leq
e_{1}|u\right]  \leq P\left[  \varepsilon_{0}\leq e_{0},\varepsilon_{1}\leq
e_{1}|u\right]  . \label{cpqd1}%
\end{equation}

Recall the example (\ref{factor}): for $d\in\left\{  0,1\right\}  ,$%
\[
\varepsilon_{d}=\rho_{d}U+\nu_{d},
\]
where $\left(  \nu_{0},\nu_{1}\right)  \perp\!\!\!\perp U.$ CPQD requires that
$\nu_{0}$ and $\nu_{1}$ be positively quadrant dependent. As a restriction on
dependence between $\varepsilon_{0}$ and $\varepsilon_{1}$ conditional on $U,$
CPQD\ has some information on the joint distribution of $Y_{0}$ and $Y_{1}$,
but not marginal distribution of $Y_{d},$ which is identified by the
distribution of $\varepsilon_{d}$ conditional on $U$ for $d\in\left\{
0,1\right\}  .$ Specifically, the lower bound on the conditional joint
distribution of $\varepsilon_{0}$ and $\varepsilon_{1}$ given $U$ improves
under CPQD as shown in (\ref{cpqd1}). This is due to the \emph{nonnegative}
sign restriction on dependence between $\varepsilon_{0}$ and $\varepsilon_{1}$
given $U$ implied by CPQD. Without CPQD, the sharp lower bound and the upper bound on
the conditional joint distribution are achieved when the conditional
distributions of $\varepsilon_{0}$ given $U$ and $\varepsilon_{1}$ given $U$
are perfectly negatively dependent and perfectly positively dependent,
respectively. Under CPQD, however, the dependence is restricted to range from
independence to perfectly positive dependence without any negative dependence.
Therefore, the lower bound under CPQD is attained when their conditional
dependence is independent.

I\ show that the lower bound on the \emph{unconditional} joint distribution
can be improved from the improved lower bound on the \emph{conditional} joint
distribution. Chebyshev's integral inequality is useful for deriving the
improved lower bound on the joint distribution of $Y_{0}$ and $Y_{1}$\ under CPQD:

\begin{description}
\item[Chebyshev's Integral Inequality] \label{L3}If $f$ and
$g:[a,b]\longrightarrow\mathbb{R}$ are two comonotonic functions, then%
\[
\frac{1}{b-a}\int\limits_{a}^{b}f\left(  x\right)  g\left(  x\right)
dx\geq\left[  \frac{1}{b-a}\int\limits_{a}^{b}f\left(  x\right)  dx\right]
\left[  \frac{1}{b-a}\int\limits_{a}^{b}g\left(  x\right)  dx\right]  .
\]

\end{description}

To establish bounds on the joint distribution, recall (\ref{jointdecompose}).
For $e_{0}=m^{-1}\left(  0,y_{0}\right)  $ and $e_{1}=m^{-1}\left(
1,y_{1}\right)  $ for $\left(  y_{0},y_{1}\right)  \in\mathbb{R}%
\times\mathbb{R},$%
\begin{align*}
&  P\left(  Y_{0}\leq y_{0},Y_{1}\leq y_{1}|0,z\right) \\
&  =P\left(  \varepsilon_{0}\leq e_{0},\varepsilon_{1}\leq e_{1}|U>p\left(
z\right)  \right)  .
\end{align*}
Now I require the additional assumption:

\begin{description}
\item[M.5] The propensity score $p(z)$ is bounded away from $0$ and $1$.
\end{description}

Under $M.5$, Chebyshev's integral inequality yields the lower bound as follows:%

\begin{align}
&  P\left(  \varepsilon_{0}\leq e_{0},\varepsilon_{1}\leq e_{1}|U>p\left(
z\right)  \right)  \label{cii}\\
&  =\frac{1}{1-p\left(  z\right)  }\int\limits_{p\left(  z\right)  }%
^{1}P\left[  \varepsilon_{0}\leq e_{0},\varepsilon_{1}\leq e_{1}|u\right]
du\nonumber\\
&  \geq\frac{1}{1-p\left(  z\right)  }\int\limits_{p\left(  z\right)  }%
^{1}P\left[  \varepsilon_{0}\leq e_{0}|u\right]  P\left[  \varepsilon_{1}\leq
e_{1}|u\right]  du\nonumber\\
&  \geq\left(  \frac{1}{1-p\left(  z\right)  }\right)  ^{2}\int
\limits_{p\left(  z\right)  }^{1}P\left[  \varepsilon_{0}\leq e_{0}|u\right]
du\int\limits_{p\left(  z\right)  }^{1}P\left[  \varepsilon_{1}\leq
e_{1}|u\right]  du.\nonumber
\end{align}
The inequality in the third line of (\ref{cii}) follows from CPQD and the
inequality in the fourth line of (\ref{cii})\ is due to Chebyshev's
integral inequality.\ Consequently, I obtain the following:%
\begin{align}
P\left(  Y\leq y_{0},Y_{1}\leq y_{1}|0,z\right)   &  \geq P\left(
y_{0}|0,z\right)  P_{1}\left(  y_{1}|0,z\right)  \label{*1}\\
&  \geq\frac{P\left(  y_{0}|0,z\right)  L_{10}^{wst}\left(  y_{1},z\right)
}{1-p\left(  z\right)  }.\nonumber
\end{align}
Similarly, the lower bound on $P\left(  Y_{0}\leq y_{0},Y\leq y_{1}%
|1,z\right)  $ is obtained as follows:
\begin{align}
P\left(  Y_{0}\leq y_{0},Y\leq y_{1}|1,z\right)   &  \geq P_{0}\left(
y_{0}|1,z\right)  P\left(  y_{1}|1,z\right)  \label{*2}\\
&  \geq\frac{L_{01}^{wst}\left(  y_{0},z\right)  P\left(  y_{1}|1,z\right)
}{p\left(  z\right)  }.\nonumber
\end{align}

Interestingly, the DTE\ is still bounded by Makarov bounds under CPQD although
the lower bound on the joint distribution improves.\ The rigorous proof is
provided in Appendix. Here I\ discuss the reason intuitively using a
graphical illustration. As shown in Figure \ref{Makarovbounds}, the DTE is a
probability corresponding to the region below the straight line $y_{1}%
=y_{0}+\delta$ and the Makarov lower bound is obtained from the rectangle
$\left\{  Y_{0}\geq y-\delta,Y_{1}\leq y\right\}  $ below the straight line
$Y_{1}=Y_{0}+\delta$ for $y\in\mathbb{R}$\ that maximizes the
Fr\'{e}chet-Hoeffding lower bound. Since the Fr\'{e}chet-Hoeffding lower bound
on $P\left(  Y_{0}\geq y-\delta,Y_{1}\leq y\right)  $ for each $y\in
\mathbb{R}$ is achieved when the joint distribution of $Y_{0}$ and $Y_{1}$
attains its upper bound, the improved lower bound on $F\left(  y_{0}%
,y_{1}\right)  $ does not affect the lower bound on the DTE. Similarly, the
Makarov upper bound is obtained from the upper bound on $1-P\left(  Y_{0}\leq
y^{\prime}-\delta,Y_{1}\geq y^{\prime}\right)  $ for $y^{\prime}\in\mathbb{R}%
$, which is in turn obtained from the Fr\'{e}chet-Hoeffding lower bound on
$P\left(  Y_{0}\leq y^{\prime}-\delta,Y_{1}\geq y^{\prime}\right)  .$
Therefore by the same token, the improved lower bound on $F\left(  y_{0}%
,y_{1}\right)  $ does not affect the upper bound on the DTE either.%
\begin{figure}
[ptb]
\begin{center}
\includegraphics[
natheight=4.728800in,
natwidth=7.604300in,
height=2.3929in,
width=3.8311in
]%
{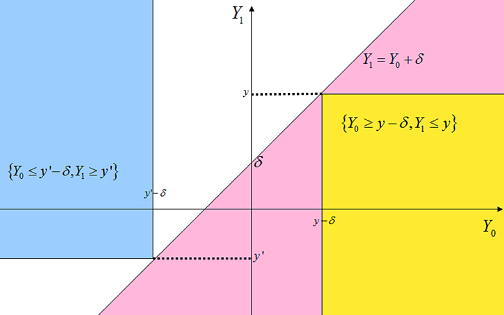}%
\caption{Makarov bounds}%
\label{Makarovbounds}%
\end{center}
\end{figure}

The specific forms of sharp bounds on marginal distributions of $Y_{0}$ and
$Y_{1},$ their joint distribution, and the DTE under $M.1-M.5$ and CPQD are
provided in Theorem \ref{T3} in Appendix. \ \ \ \ \ \ \ \ \ \ \ \ \ \ \ \ \ \ \ \ \ \ \ \ \ \ \ \ \ \ \ \ \ \ \ \ \ \ \ \ \ \ \ \ \ \ \ \ \ \ \ \ \ \ \ \ \ \ \ \ \ \ \ \ \ \ \ \ \ \ \ \ \ \ \ \ \ \ \ \ \ \ \ \ \ \ \ \ \ \ \ \ \ \ \ \ \ \ \ \ \ \ \ \ \ \ \ \ \ \ \ \ \ \ \ \ \ \ \ \ \ \ \ \ \ \ \ \ \ \ \ \ \ \ \ \ \ \ \ \ \ \ \ \ \ \ \ \ \ \ \ \ \ \ \ \ \ \ \ \ \ \ \ \ \ \ \ \ \ \ \ \ \ \ \ \ \ \ \ \ \ \ \ \ \ \ \ \ \ \ \ \ \ \ \ \ \ \ \ \ \ \ \ \ \ \ \ \ \ \ \ \ \ \ \ \ \ \ \ \ \ \ \ \ \ \ \ \ \ \ \ \ \ \ \ \ \ \ \ \ \ \ \ \ \ \ \ \ \ \ \ \ \ \ \ \ \ \ \ \ \ \ \ \ \ \ \ \ \ \ \ \ \ \ \ \ \ \ \ \ \ \ \ \ \ \ \ \ \ \ \ \ \ \ \ \ \ \ \ \ \ \ \ \ \ \ \ \ \ \ \ \ \ \ \ \ \ \ \ \ \ \ \ \ \ \ \ \ \ \ \ \ \ \ \ \ \ \ \ \ \ \ \ \ \ \ \ \ \ \ \ \ \ \ \ \ \ \ \ \ \ \ \ \ \ \ \ 

\subsection{Monotone Treatment Response}
\label{c2s3s4}
In this subsection, I\ maintain $M.1-M.4$ on the model (\ref{model}) and
additionally impose MTR, which is written as $P\left(  Y_{1}\geq Y_{0}\right)
=1$. As illustrated in Figure \ref{MTRf}, MTR is a restriction imposed on the
support of $(Y_{0},Y_{1})$, while NSM\ and CPQD directly restrict the sign of
dependence between unobservables. I\ show that MTR has substantial
identifying power for the marginal distributions, the joint distribution, and
the DTE.%

\begin{figure}
[ptb]
\begin{center}
\includegraphics[
natheight=2.864300in,
natwidth=2.791600in,
height=1.4607in,
width=1.4243in
]%
{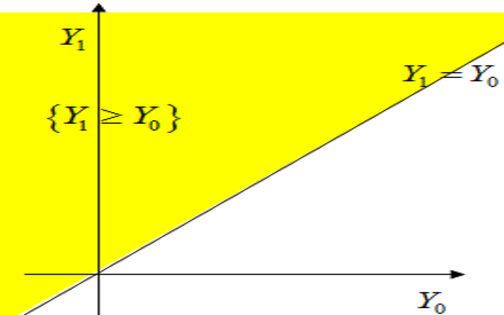}%
\caption{Support under MTR}%
\label{MTRf}%
\end{center}
\end{figure}

Start with bounds on marginal distributions. Remember that NSM\ as well as
$M.1-M.4$ has no additional identifying power for the upper bound on $P_{1}\left(
y,0|z\right)  $ and the lower bound on $P_{0}\left(  y,1|z\right)  $.
Interestingly, MTR improves both the upper bound on $P_{1}\left(
y,0|z\right)  $ and the lower bound on $P_{0}\left(  y,1|z\right)  .$ On the
other hand, unlike NSM, MTR does not have any identifying power on the lower
bound on $P_{1}\left(  y,0|z\right)  $ and the upper bound on $P_{0}\left(
y,1|z\right)  .$ Recall that in (\ref{1.1}),
\begin{align*}
&  P_{1}\left(  y,0|z\right)  \\
&  =P\left(  Y_{1}\leq y,p\left(  z\right)  <U\leq\overline{p}\right)
+P\left(  Y_{1}\leq y|\overline{p}<U\right)  \left(  1-\overline{p}\right)  .
\end{align*}
Since MTR implies stochastic dominance of $Y_{1}$ over $Y_{0}$, under MTR,%
\[
P\left(  Y_{1}\leq y|\overline{p}<U\right)  \leq P\left(  Y_{0}\leq
y|\overline{p}<U\right)  =\underset{p\left(  z\right)  \rightarrow\overline
{p}}{\lim}P\left(  y|0,z\right)  .
\]
Similarly,%
\[
P\left(  Y_{0}\leq y|U\leq\underline{p}\right)  \geq P\left(  Y_{1}\leq
y|U\leq\underline{p}\right)  =\underset{p\left(  z\right)  \rightarrow
\underline{p}}{\lim}P\left(  y|1,z\right)  .
\]
This shows that\ MTR tightens the upper bound on $P_{1}\left(  y,0|z\right)  $
and the lower bound on $P_{0}\left(  y,1|z\right)  $.

\begin{lemma}
\label{L4}Under $M.1-M.4$ and MTR, $P_{1}\left(  y,0|z\right)  $ and
$P_{0}\left(  y,1|z\right)  \ $are bounded as follows:%
\begin{align*}
P_{1}\left(  y,0|z\right)   &  \in\left[  L_{10}^{wst}\left(  y,z\right)
,U_{10}^{mtr}\left(  y,z\right)  \right]  ,\\
P_{0}\left(  y,1|z\right)   &  \in\left[  L_{01}^{mtr}\left(  y,z\right)
,U_{01}^{wst}\left(  y,z\right)  \right]  ,
\end{align*}
where%
\begin{align*}
L_{01}^{mtr}\left(  y,z\right)   &  =\underset{p\left(  z\right)
\rightarrow\underline{p}}{\lim}P\left(  y|0,z\right)  \left(  1-\underline
{p}\right)  -P\left(  y|0,z\right)  \left(  1-p\left(  z\right)  \right)
+\underset{p\left(  z\right)  \rightarrow\underline{p}}{\lim}P\left(
y|1,z\right)  \underline{p},\\
U_{10}^{mtr}\left(  y,z\right)   &  =\underset{p\left(  z\right)
\rightarrow\overline{p}}{\lim}P\left(  y|1,z\right)  \overline{p}-P\left(
y|1,z\right)  p\left(  z\right)  +\underset{p\left(  z\right)  \rightarrow
\overline{p}}{\lim}P\left(  y|0,z\right)  \left(  1-\overline{p}\right)  ,
\end{align*}
and these bounds are sharp.
\end{lemma}

From Lemma \ref{L4}, sharp bounds on marginal distributions of $Y_{0}$ and
$Y_{1}$ are improved based on $L_{01}^{mtr}\left(  y,z\right)  $ and
$U_{10}^{mtr}\left(  y|z\right)  $ under $M.1-M.4$, and MTR as follows:%
\begin{align*}
F_{0}^{L}\left(  y\right)   &  =\underset{z\in\mathcal{Z}}{\sup}\left[  P\left(
y|0,z\right)  \left(  1-p\left(  z\right)  \right)  +L_{01}^{mtr}\left(
y,z\right)  \right]  ,\\
F_{0}^{U}\left(  y\right)   &  =\underset{z\in\mathcal{Z}}{\inf}\left[  P\left(
y|0,z\right)  \left(  1-p\left(  z\right)  \right)  +U_{01}^{wst}\left(
y,z\right)  \right]  ,\\
F_{1}^{L}\left(  y\right)   &  =\underset{z\in\mathcal{Z}}{\sup}\left[  P\left(
y|1,z\right)  p\left(  z\right)  +L_{10}^{wst}\left(  y,z\right)  \right]  ,\\
F_{1}^{U}\left(  y\right)   &  =\underset{z\in\mathcal{Z}}{\inf}\left[  P\left(
y|1,z\right)  p\left(  z\right)  +U_{10}^{mtr}\left(  y,z\right)  \right]  .
\end{align*}%
\begin{figure}
[ptb]
\begin{center}
\includegraphics[
natheight=3.072700in,
natwidth=2.802000in,
height=1.5645in,
width=1.4287in
]%
{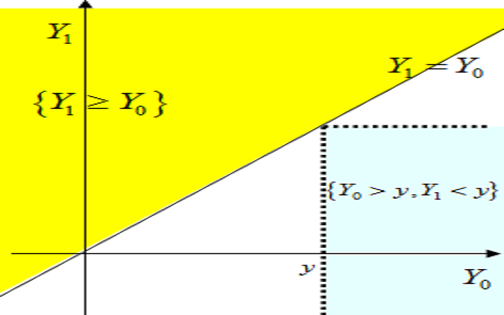}%
\caption{$P\left(  Y_{0}>Y_{1}\right)  =P\left[  \underset{y\in\mathbb{R}%
}{\cup}\left\{  Y_{0}>y,Y_{1}<y\right\}  \right]  $}%
\label{MTRf2}%
\end{center}
\end{figure}

Now, I\ show that MTR also has identifying power for the joint distribution.
I\ will use Lemma \ref{L4.6} to bound the joint distribution under MTR.
Henceforth, $x^{+}$ denotes $max\left(  x,0\right)  .$

\begin{lemma}
\label{L4.6}(\citet{N2006}) Suppose that marginal distributions $F_{0}$ and
$F_{1}$ are known and that $F\left(  a_{0},a_{1}\right)  =\theta$ where
$\left(  a_{0},a_{1}\right) \in \mathbb{R}^{2}$ and $\theta$ satisfies
$\max\left(  F_{0}\left(  a_{0}\right)  +F_{1}\left(  a_{1}\right)
-1,0\right)  \leq\theta\leq\min\left(  F_{0}\left(  a_{0}\right)
,F_{1}\left(  a_{1}\right)  \right)  .$ Then, sharp bounds on the joint
distribution $F$ are given as follows:%
\[
F^{L}\left(  y_{0},y_{1}\right)  \leq F\left(  y_{0},y_{1}\right)  \leq
F^{U}\left(  y_{0},y_{1}\right)  ,
\]
where%
\begin{align*}
F^{L}\left(  y_{0},y_{1}\right)   &  =\max\left\{  0,F_{0}\left(
a_{0}\right)  +F_{1}\left(  a_{1}\right)  -1,\theta-\left(  F_{0}\left(
a_{0}\right)  -F_{0}\left(  y_{0}\right)  \right)  ^{+}-\left(  F_{1}\left(
a_{1}\right)  -F_{1}\left(  y_{1}\right)  \right)  ^{+}\right\}  ,\\
F^{L}\left(  y_{0},y_{1}\right)   &  =\min\left\{  F_{0}\left(  y_{0}\right)
,F_{1}\left(  y_{1}\right)  ,\theta+\left(  F_{0}\left(  y_{0}\right)
-F_{0}\left(  a_{0}\right)  \right)  ^{+}+\left(  F_{1}\left(  y_{1}\right)
-F_{1}\left(  a_{1}\right)  \right)  ^{+}\right\}  .
\end{align*}

\end{lemma}

Suppose that marginal distributions $F_{0}$ and $F_{1}$ are fixed. Lemma \ref{L4.6}
shows that sharp bounds on the joint distribution improve when the values of
the joint distribution are known at some fixed points. Note that $P\left(
Y_{1}\geq Y_{0}\right)  =1$ if and only if $F\left(  y,y\right)  =F_{1}\left(
y\right)  $ for all $y\in\mathbb{R}.$ As illustrated in Figure \ref{MTRf2},
\[
P\left(  Y_{0}>Y_{1}\right)  =P\left[  \underset{y\in\mathbb{R}}{\cup}\left\{
Y_{0}>y,Y_{1}<y\right\}  \right]  .
\]
Therefore,

\begin{center}
$%
\begin{array}
[c]{cl}
& P\left(  Y_{1}\geq Y_{0}\right)  =1\\
\Longleftrightarrow & P\left(  Y_{0}>Y_{1}\right)  =0\\
\Longleftrightarrow & P\left(  Y_{0}>y,Y_{1}<y\right)  =0\text{ for all }%
y\in\mathbb{R}\\
\Longleftrightarrow & F\left(  y,y\right)  =F_{1}\left(  y\right)  ,\text{ for
all }y\in\mathbb{R}.
\end{array}
$
\end{center}

Since for each $y\in\mathbb{R}$ the value of $F\left(  y,y\right)
$\ is\ known from the fixed marginal distribution $F_{1}$ under MTR, sharp bounds on the joint distribution can be derived by
taking the intersection of the bounds under the restriction $F\left(
y,y\right)  =F_{1}\left(  y\right)  $ over all $y$ $\in\mathbb{R}$. Technical
details are presented in Appendix.%

\begin{figure}
[ptb]
\begin{center}
\includegraphics[
natheight=9.479200in,
natwidth=15.167100in,
height=2.1153in,
width=2.7985in
]%
{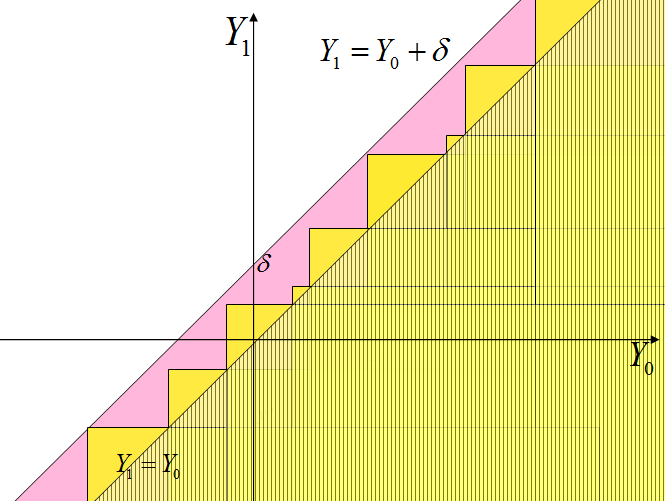}%
\caption{Improved lower bound on the DTE under MTR}%
\label{improvedlowerMTR}%
\end{center}
\end{figure}

In Chapter 1, I obtained sharp bounds on the DTE when marginal distributions
are fixed and MTR is imposed. Compared to Figure \ref{Makarovbounds}, Figure
\ref{improvedlowerMTR} shows that under MTR the lower bound on the DTE
improves by allowing more mass to be added between $Y_{1}=Y_{0}+\delta$ and
$Y_{1}=Y_{0}$.  Lemma \ref{L5} presents sharp bounds on the DTE under MTR and
fixed marginals $F_{0}$ an $F_{1}$ as follows:

\begin{lemma}
\label{L5} (\citet{K2014}) Under MTR, sharp bounds on the DTE are given as
follows: for fixed marginals $F_{0}$ an $F_{1}$ and any $\delta\in\mathbb{R},$%
\[
F_{\Delta}^{L}\left(  \delta\right)  \leq F_{\Delta}\left(  \delta\right)
\leq F_{\Delta}^{U}\left(  \delta\right)  ,
\]
where%
\begin{align*}
F_{\Delta}^{U}\left(  \delta\right)   &  =\left\{
\begin{array}
[c]{cc}%
1+\underset{y\in\mathbb{R}}{\inf}\left\{  \min\left(  F_{1}\left(  y\right)
-F_{0}\left(  y-\delta\right)  \right)  ,0\right\}  , & \text{for }\delta
\geq0,\\
0, & \text{for }\delta<0.
\end{array}
\right.  ,\\
F_{\Delta}^{L}\left(  \delta\right)   &  =\left\{
\begin{array}
[c]{cc}%
\underset{\left\{  a_{k}\right\}  _{k=-\infty}^{\infty}\in\mathcal{A}_{\delta
}}{\sup}\sum\limits_{k=-\infty}^{\infty}\max\left\{  F_{1}\left(
a_{k+1}\right)  -F_{0}\left(  a_{k}\right)  ,0\right\}  , & \text{for }%
\delta\geq0,\\
0, & \text{for }\delta<0,
\end{array}
\right. \\
\text{where }\mathcal{A}_{\delta}  &  =\left\{  \left\{  a_{k}\right\}
_{k=-\infty}^{\infty};0\leq a_{k+1}-a_{k}\leq\delta\text{ for every integer
}k\right\}  .
\end{align*}

\end{lemma}

From Lemmas \ref{L4}, \ref{L4.6}, and \ref{L5}, it is straightforward to derive
sharp bounds on the joint distribution and the DTE under $M.1-M.4$ and MTR.

The specific forms of sharp bounds on marginal distributions of $Y_{0}$ and
$Y_{1},$ their joint distribution, and the DTE under $M.1-M.4$ and MTR are
provided in Theorem \ref{T4} in Appendix.

\section{Discussion}
\label{c2s4}
\subsection{Testable Implications}
\label{c2s4s1}
I\ here show that NSM and MTR yield testable implications.

Note that NSM implies the following: for any $\left(  z^{\prime},z\right)
\in\mathcal{Z}\times\mathcal{Z}$ such that $p\left(  z^{\prime}\right)  \geq p\left(
z\right)  ,$ and for any $y\in\mathbb{R}$,
\begin{align*}
P\left(  \varepsilon_{1}\leq m^{-1}\left(  1,y\right)  |U\leq p\left(
z\right)  \right)   &  \leq P\left(  \varepsilon_{1}\leq m^{-1}\left(
1,y\right)  |U\leq p\left(  z^{\prime}\right)  \right)  ,\\
P\left(  \varepsilon_{0}\leq m^{-1}\left(  0,y\right)  |U>p\left(  z\right)
\right)   &  \leq P\left(  \varepsilon_{0}\leq m^{-1}\left(  0,y\right)
|U>p\left(  z^{\prime}\right)  \right)  .
\end{align*}
\newline This yields the following testable form of functional inequalities:
\begin{align}
P\left(  y|1,z\right)   &  \leq P\left(  y|1,z^{\prime}\right)
,\label{fcntest}\\
P\left(  y|0,z\right)   &  \leq P\left(  y|0,z^{\prime}\right)  .\nonumber
\end{align}

Next, MTR has two testable implications. First, MTR implies stochastic
dominance. In our model, marginal distributions are partially identified for
the entire population. Therefore, it can be tested by applying econometric
techniques for testing stochastic dominance for partially identified marginal
distributions as proposed in the literature including Jun et al. (2013). Also,
the sharp lower bound on the DTE under MTR can be greater than the upper bound
and furthermore the lower bound could be even above 1, when MTR is violated
for the true joint distribution of $Y_{0}$ and $Y_{1}$.

\subsection{NSM+CPQD and NSM+MTR}
\label{c2s4s2}
In Section \ref{c2s3}, I\ explored the identifying power of NSM, CPQD, and
MTR, separately. In this subsection, I\ briefly discuss how sharp bounds are
constructed when some of these conditions are combined. Establishing sharp
bounds under NSM and CPQD and sharp bounds under NSM and MTR is straightforward
from the results in Subsection \ref{c2s3s2} - Subsection \ref{c2s3s4}. First,
under NSM and CPQD, bounds on marginal distributions and bounds on the DTE are
identical to those under NSM only, since CPQD has no identifying power on
the marginal distributions and the DTE. The bounds on the joint distribution
under NSM and CPQD can be established by plugging the bounds on the
counterfactual probabilities $P_{0}\left(y_{0},1|z\right)  $ and
$P_{1}\left(y_{1},0|z\right)  $\ under NSM into the upper bound formula
under CPQD as follows:%
\begin{align*}
F^{L}\left(  y_{0},y_{1}\right)   &  =\underset{z\in\mathcal{Z}}{\sup}\left\{
P\left(  y_{0}|0,z\right)  L_{10}^{sm}\left(  y_{1},z\right)  +L_{01}%
^{wst}\left(  y_{0},z\right)  P\left(  y_{1}|1,z\right)  \right\}  ,\\
F^{U}\left(  y_{0},y_{1}\right)   &  =\inf_{z\in\mathcal{Z}}\left[
\begin{array}
[c]{c}%
\min\left\{  P\left(  y_{0}|0,z\right)  \left(  1-p\left(  z\right)  \right)
,U_{10}^{wst}\left(  y,z\right)  \right\} \\
+\min\left\{  U_{01}^{sm}\left(  y_{0},z\right)  ,P\left(  y_{1}|1,z\right)
p\left(  z\right)  \right\}
\end{array}
\right]  .
\end{align*}

Similarly, the distributional parameters are bounded under NSM and MTR. The
specific forms of sharp bounds on marginal distributions of $Y_{0}$ and
$Y_{1},$ their joint distribution, and the DTE under $M.1-M.4,$ NSM, and MTR
are provided in Theorem \ref{T5} in Appendix.

Lastly, marginal distribution bounds under NSM, CPQD, and MTR and marginal
distribution bounds under CPQD and MTR are identical to those under NSM and MTR
and those under MTR, respectively, since CPQD does not affect bounds on
marginal distributions. However, it is not straightforward to construct sharp
bounds on the joint distribution and the DTE under these three conditions or
under CPQD and MTR, as both CPQD and MTR directly restrict the joint
distribution as different types of conditions. To the best of my knowledge,
there exist no results on the sharp bounds on the joint distribution and DTE
when support restrictions such as MTR are combined with various dependence
restriction such as quadrant dependence. This is beyond the scope of this paper.

\section{Numerical Examples}
\label{c2s5}
This section presents numerical examples to illustrate how bounds on
distributional parameters are tightened by the restrictions considered in this
paper. The potential outcomes and selection equations are given as follows:%
\begin{align*}
Y_{0}  &  =\rho U+\varepsilon,\\
Y_{1}  &  =Y_{0}+\eta,\\
D\left(  Z\right)   &  =1\left(  Z\geq U\right)  ,
\end{align*}
where $\left(  U,\varepsilon\right)  \sim i.i.d.N\left(  0,I_{2}\right)  $,
$\eta\sim\chi^{2}\left(  k\right)  $, and $\eta\perp\!\!\!\perp\left(
U,\varepsilon\right)  $ \ for a positive integer $k.$

Selection is allowed to be endogenous since the selection unobservable $U$
\ is dependent on potential outcomes $Y_{0}$ and $Y_{1}$ for $\rho\neq0$.
I\ consider negative values of $\rho$\ to make the specification satisfy NSM
discussed in Subsection \ref{c2s3s2}. CPQD holds due to the common factor
$\varepsilon$ in $Y_{0}$ and $Y_{1}$, which is independent of $U$. Lastly, MTR
is obviously satisfied as $P\left(  Y_{1}\geq Y_{0}\right)  =1$ since
$\eta\geq0$ with probability one. Also, to rule out the full support of the
instrument, $Z$ is assumed to be a uniformly distributed random variable on
$\left(  z,-z\right)  $\ for $z=2,1.5,1,.5.$%

\begin{figure}
[ptb]
\begin{center}
\includegraphics[
height=2.8591in,
width=6.077in
]%
{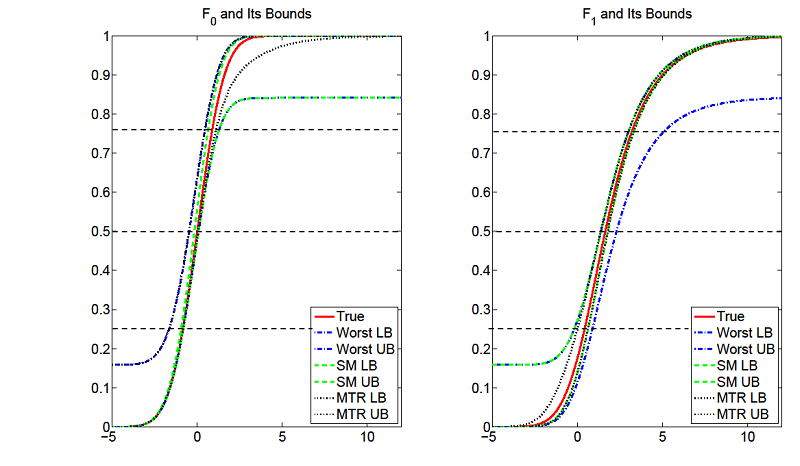}%
\caption{Bounds on the distributions of $Y_{0}$ (left) and $Y_{1}$ (right)}%
\label{marginalbounds}%
\end{center}
\end{figure}

First, for $\rho=-0.75$ and $Z\sim Unif\left(  1,-1\right)  ,$ I obtain the
sharp bounds on the marginal distributions of potential outcomes $Y_{0}$ and
$Y_{1}$ as proposed in Section \ref{c2s3}. Figure \ref{marginalbounds} shows
the bounds on each potential outcome distribution as well as the true
distribution. Solid curves represent the true marginal distributions of
$Y_{0}$ and $Y_{1}$ and dash-dot curves, dotted curves, and dashed curves
represent their worst bounds, bounds under NSM, and bounds under MTR,
respectively. Remember that bounds on marginal distributions under CPQD are
identical to worst bounds. Figure \ref{marginalbounds} shows that NSM
substantially improves the upper bound on $F_{0}$ and the lower bound on
$F_{1}$, compared to worst bounds. As shown in Lemma \ref{L2}, NSM improves the
upper bound on $P\left(  Y_{0}\leq y,1|z\right)  $\ and the lower bound on
$P\left(  Y_{1}\leq y,0|z\right)  $ for $y\in\mathbb{R},$ which are used in
obtaining the upper bound on $F_{0}$ and the lower bound on $F_{1},$
respectively. On the other hand, MTR improves the lower bound on $F_{0}$ and
the upper bound on $F_{1}$. \ Note that in contrast to NSM, MTR improves the
lower bound on $P\left(  Y_{0}\leq y,1|z\right)  $\ and the upper bound on
$P\left(  Y_{1}\leq y,0|z\right)  $ for all $y\in\mathbb{R},$ which are used
in obtaining the lower bound on $F_{0}$ and the upper bound on $F_{1},$ respectively.%

\begin{figure}
[ptb]
\begin{center}
\includegraphics[
height=2.8461in,
width=5.1007in
]%
{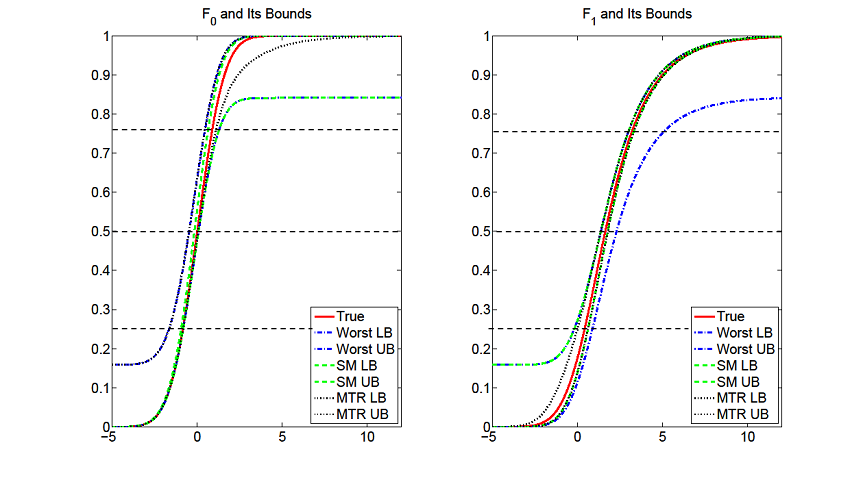}%
\caption{Bounds on the distributions of $Y_{0}$ (left) and $Y_{1}$ (right)}%
\label{marginalbounds_joint}%
\end{center}
\end{figure}

Next, I\ plotted bounds on marginal distributions when NSM and MTR are jointly
imposed. In Figure \ref{marginalbounds_joint}, solid curves represent the true
distributions of $Y_{0}$ and $Y_{1},$ and dash-dot curves and dashed curves
represent their worst bounds and bounds under NSM and MTR, respectively. Figure
\ref{marginalbounds_joint} shows that if NSM and MTR are jointly considered,
both upper and lower bounds improve for both $F_{0}$ and $F_{1}$ as
discussed in Section \ref{c2s4}. The quantiles of the potential outcomes can
be obtained by inverting the bounds on the marginal distributions. The bounds
on the quantiles of $Y_{0}$ and $Y_{1}$ are reported in Table \ref{Table0}

\begin{figure}
[ptb]
\begin{center}
\includegraphics[
height=3.4307in,
width=6.1505in
]%
{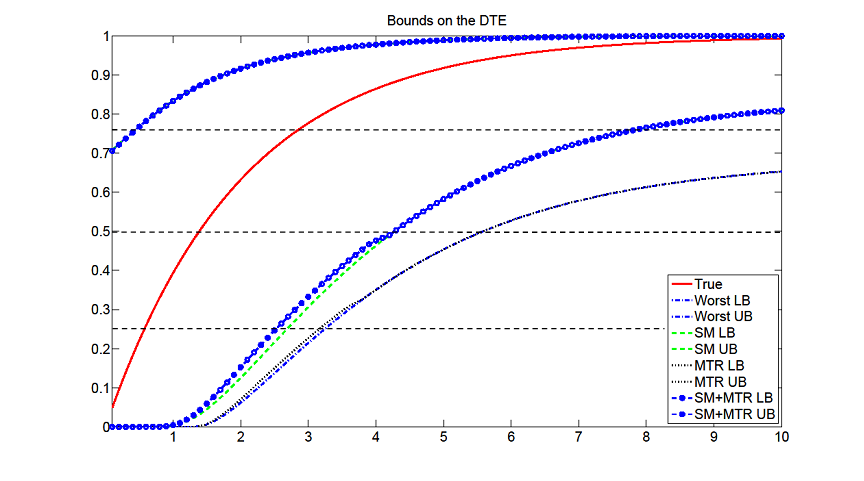}%
\caption{True DTE and bounds on the DTE}%
\label{dtebounds_joint}%
\end{center}
\end{figure}

Figure \ref{dtebounds_joint} shows the true DTE and bounds on the DTE. Solid
curve, dash-dot curves, dotted lines, dashed curves, and \ dashed curves with
circles represent the true DTE, worst DTE bounds, bounds under NSM, bounds
under MTR, and bounds under NSM and MTR, respectively. Compared to the worst
bounds, the lower bound under NSM notably improves over the entire support of
the DTE. Remember that the lower DTE bound improves through the upper bound on
$P_{0}\left(  y,1|z\right)  $\ and the lower bound on $P_{1}\left(
y,0|z\right)  ,$ both of which are improved by NSM, even though the DTE bounds
under NSM still relies on Makarov bounds. On the other hand, although MTR
directly improves the lower DTE bound from the Makarov lower bound, the
improvement of the lower DTE bound by MTR is not substantial over the whole
support. This is because neither the upper bound on $P_{0}\left(
y,1|z\right)  $\ nor the lower bound on $P_{1}\left(  y,0|z\right)  $
improves, which are the counterfactual components consisting of the lower
bound. Also, as discussed in Chapter 1, the sharp lower bound on
$F_{\Delta}\left(  \delta\right)  $ under MTR converges to the Makarov lower
bound as $\delta$ increases for sufficiently large values of $\delta.$ On the
other hand, the upper bound under NSM does not improve from the worst upper
bound as discussed in Subsection \ref{c2s3s2} Although the upper bound improves under
MTR through improvement in the lower bound on $P_{0}\left(  y,1|z\right)
$\ and the upper bound on $P_{1}\left(  y,0|z\right)  $, the improvement in
the upper bound under MTR is not remarkable as shown in Figure
\ref{dtebounds_joint}. Also, the quantiles of treatment effects can be
obtained by inverting the bounds on the DTE. The bounds on the quantiles of
the DTE are reported in Table \ref{Table0}.

Table \ref{Table1} shows the bounds on the joint distribution under various
restrictions considered in this study. Compared to the worst bounds, bounds
are tighter under NSM due to the marginal distributions bounds improved by NSM.
On the other hand, the upper bound under CQPD does not improve unlike the
lower bound. Note that CQPD has no identifying power on marginal
distributions, while it improves the lower bound on the joint distribution.
However, when CQPD is combined with NSM, the upper bound also improves due to
the improved marginal distributions bounds under NSM. The identification region
under MTR is tighter than the worst identification region for both the upper
bound and the lower bound. Note that the upper bound under MTR\ is lower than
the worst lower bound through the improved lower bound on $P_{0}\left(
y,1|z\right)  $ and improved upper bound on $P_{1}\left(y,0|z\right)
$ by MTR, while it still poses the Makarov upper bound. On the other hand, the
lower bound under MTR is higher than the worst lower bound obtained from the
Makarov lower bound because of the direct effect of MTR on the lower bound on
the joint distribution. Remember that the lower bound on the joint
distribution is not affected by the improved components of the bounds on
counterfactual probabilities: the improved lower bound on $P_{0}\left(y,1|z\right)  $ and improved upper bound on $P_{1}\left(y,0|z\right)
$. Lastly, under NSM\ and MTR both the lower bound and the upper bound improve
through counterfactual probabilities $U_{01}^{sm}(y,z)$ and$\ L_{10}%
^{sm}(y,z)$, respectively which are improved by NSM compared to the bounds
under MTR only.

I also obtained sharp bounds on the potential outcomes distributions and the
DTE for $z\in\left\{  2,1.5,1,.5\right\}  $ to see how the support of the
instrument affect the identification region. Tables \ref{Table2},
\ref{Table3}, and \ref{Table4} document the identification regions of $F_{0},$
$F_{1},$ and $F_{\Delta},$ respectively, under NSM and MTR for these different
values of $z.$ As expected, as the support of the instrument gets larger, the
identification regions of the marginal distributions and the DTE become more
informative. Table 5 shows the identification regions of the DTE for different
values of $\rho=\left\{  -.25,-.5,-.75\right\}  $. Since the true DTE does not
depend on the value of $\rho,$ one can see from Table 5 how the size of
correlation between the outcome heterogeneity and the selection heterogeneity
affects the identification region of the DTE for the fixed true DTE. As shown
in Table 5, the identification region becomes tighter\ as $\rho$ approaches
$0$. That is, the weaker endogeneity with the smaller absolute value of $\rho$
helps identification of the DTE. This is readily understood from the extreme
case. If $\rho=0$ where the treatment selection is independent of potential
outcomes $Y_{0}$ and $Y_{1},$ marginal distributions of potential outcomes are
exactly identified, which clearly leads to tighter bounds on the DTE.

\section{Conclusion}
\label{c2s6}
In this paper, I established sharp bounds on marginal distributions of
potential outcomes, their joint distribution, and the DTE in triangular
systems. To do this, I explored various types of restrictions to tighten the
existing bounds including stochastic monotonicity between each outcome
unobservable and the selection unobservable, conditional positive quadrant
dependence between two outcome unobservables given the selection unobservable,
and the monotonicity of the potential outcomes. I did not rely on rank
similarity and the full support of IV, and furthermore I avoided strong
distributional assumptions including a single factor structure, which
contrasts with most of related work. The proposed bounds take the form of
intersection bounds and lend themselves to existing inference methods
developed in \citet{CLR2013}. 

\begin{table}[tbp]%

\begin{center}
\caption{True quantiles and bounds on the quantiles of $Y_{0}$ and $Y_{1}$}\label{Table0}%
\end{center}

\centering%
\begin{tabular}
[c]{c|c|ccc}\hline\hline
$q$ &  & $F_{0}^{-1}\left(  q\right)  $ & $F_{1}^{-1}\left(  q\right)  $ &
$F_{\Delta}^{-1}\left(  q\right)  $\\\hline
$.25$ & True & $-.85$ & $.40$ & $.48$\\
& Worst & $[-1.70,-.85]$ & $[-.20,.90]$ & $[0,3.15]$\\
& NSM & $[-.95,-.85]$ & $[-.20,.60]$ & $[0,2.60]$\\
& MTR & $[-1.70,-.85]$ & $[0,.60]$ & $[0,3.05]$\\
& NSM+MTR & $[-.95,-.85]$ & $[0,.60]$ & $[0,2.40]$\\\hline
$.5$ & True & $0$ & $1.65$ & $1.30$\\
& Worst & $[-.45,.05]$ & $[0,2.30]$ & $[0,5.50]$\\
& NSM & $[-.15.05]$ & $[1.40,1.80]$ & $[0,4.20]$\\
& MTR & $[-.45,.05]$ & $[1.40,1.80]$ & $[0,5.50]$\\
& NSM+MTR & $[-.15,.05]$ & $[1.40,1.80]$ & $[0,4.20]$\\\hline
$.75$ & True & $.85$ & $3.15$ & $2.70$\\
& Worst & $[.40,1.20]$ & $[2.95,4.95]$ & $[.25,\infty)$\\
& NSM & $[.60,1.20]$ & $[2.95,3.30]$ & $[.25,7.40]$\\
& MTR & $[.40,1.05]$ & $[2.95,3.30]$ & $[.25,\infty)$\\
& NSM+MTR & $[.60,1.05]$ & $[2.95,3.30]$ & $[.25,7.40]$\\\hline
\end{tabular}%
\end{table}%

\pagebreak
\newpage
\begin{table}[H]%

\begin{center}
\caption{True Joint distribution $F\left( y_{0},y_{1}\right)$ and its bounds under various restrictions
}\label{Table1}
\end{center}

\centering
\begin{tabular}
[c]{c|c|ccccccc}\hline\hline
$y_{0}\backslash y_{1}$ &  & $-3$ & $-1$ & $1$ & $3$ & $5$ & $7$ & $9$\\\hline

$-1$ & True & $0$ & $.03$ & $.16$ & $.19$ & $.20$ & $.21$ & $.21$\\
& Worst & $[0,.09]$ & $[0,.12]$ & $[0,.23]$ & $[0,.30]$ & $[.03,.34]$ &
$[.09,.36]$ & $[.11,.36]$\\
& NSM & $[0,.09]$ & $[0,.12]$ & $[0,.23]$ & $[0,.24]$ & $[.13,.24]$ &
$[.18,.24]$ & $[.20,.24]$\\
& CPQD & $[0,.09]$ & $[.01,.12]$ & $[.06,.23]$ & $[.13,.30]$ & $[.15,.34]$ &
$[.16,.36]$ & $[.17,.36]$\\
& NSM+CPQD & $[0,.09]$ & $[.01,.12]$ & $[.08,.23]$ & $[0.16,.24]$ & $[.19,.24]
$ & $[.20,.24]$ & $[.21,.24]$\\
& MTR & $[0,.01]$ & $[.03,.12]$ & $[.03,.23]$ & $[.03,.30]$ & $[.03,.34]$ &
$[.09,.36]$ & $[.11,.36]$\\
& NSM+MTR & $[0,.01]$ & $[.03,.12]$ & $[.03,.23]$ & $[.03.24]$ & $[.13,.24]$ &
$[.18,.24]$ & $[.20,.24]$\\\hline
$1$ & True & $0$ & $.03$ & $.37$ & $.63$ & $.73$ & $.77$ & $.78$\\
& Worst & $[0,.16]$ & $[0,.18]$ & $[.13,.43]$ & $[.38,.75]$ & $[.50,.85]$ &
$[.54,.87]$ & $[.55,.87]$\\
& NSM & $[0,.16]$ & $[0,.18]$ & $[.19,.43]$ & $[.50,.75]$ & $[.64,.85]$ &
$[.69,.85]$ & $[.71,.85]$\\
& CPQD & $[0,.16]$ & $[.02,.18]$ & $[.21,.43]$ & $[.43,.75]$ & $[.53,.85]$ &
$[.56,.87]$ & $[.58,.87]$\\
& NSM+CPQD & $[0,.16]$ & $[.03,.18]$ & $[.26,.43]$ & $[.53,.75]$ & $[.65,.85]$
& $[.69,.85]$ & $[.71,.85]$\\
& MTR & $[0,.01]$ & $[.04,.12]$ & $[.33,.43]$ & $[.39,.75]$ & $[.50,.85]$ &
$[.55,.87]$ & $[.57,.87]$\\
& NSM+MTR & $[0,.01]$ & $[.04,.12]$ & $[.33,.43]$ & $[.50,.75]$ & $[.64,.85]$ &
$[.70,.85]$ & $[.73,.85]$\\\hline
$3$ & True & $0$ & $.03$ & $.37$ & $.75$ & $.90$ & $.96$ & $.98$\\
& Worst & $[0,.16]$ & $[.03,.19]$ & $[.25,.43]$ & $[.51,.76]$ & $[.62,.91]$ &
$[.66,.97]$ & $[.67,.99]$\\
& NSM & $[0,.16]$ & $[.03,.19]$ & $[.31,.43]$ & $[.62,.76]$ & $[.76,.91]$ &
$[.81,.97]$ & $[.83,.99]$\\
& CPQD & $[0,.16]$ & $[.03,.19]$ & $[.25,.43]$ & $[.51,.76]$ & $[.62,.91]$ &
$[.66,.97]$ & $[.67,.99]$\\
& NSM+CPQD0 & $[0,.16]$ & $[.03,.19]$ & $[.31,.43]$ & $[.63,.76]$ & $[.76,.91]
$ & $[.81,.97]$ & $[.83,.99]$\\
& MTR & $[0,.01]$ & $[.04,.12]$ & $[.33,.43]$ & $[.72,.76]$ & $[.68,.91]$ &
$[.74,.97]$ & $[.76,.99]$\\
& NSM+MTR & $[0,.01]$ & $[.04,.12]$ & $[.33,.43]$ & $[.72,.76]$ & $[.82,.91]$ &
$[.89,.97]$ & $[.92,.99]$\\\hline
$5$ & True & $0$ & $.03$ & $.37$ & $.75$ & $.90$ & $.96$ & $.98$\\
& Worst & $[0,.16]$ & $[.03,.19]$ & $[.25,.43]$ & $[.51,.76]$ & $[.62,.91]$ &
$[.66,.97]$ & $[.68,.99]$\\
& NSM & $[0,.16]$ & $[.03,.19]$ & $[.31,.43]$ & $[.63,.76]$ & $[.76,.91]$ &
$[.81,.97]$ & $[.83,.99]$\\
& CPQD & $[0,.16]$ & $[.03,.19]$ & $[.25,.43]$ & $[.51,.76]$ & $[.62,.91]$ &
$[.66,.97]$ & $[.68,.99]$\\
& NSM+CPQD & $[0,.16]$ & $[.03,.19]$ & $[.31,.43]$ & $[.63,.76]$ & $[.76,.91]$
& $[.81,.97]$ & $[.83,.99]$\\
& MTR & $[0,.01]$ & $[.04,.12]$ & $[.33,.43]$ & $[.72,.76]$ & $[.90,.91]$ &
$[.94,.97]$ & $[.96,.99]$\\
& NSM+MTR & $[0,.01]$ & $[.04,.12]$ & $[.33,.43]$ & $[.72,.76]$ & $[.90,.91]$ &
$[.94,.97]$ & $[.96,.99]$\\\hline
$7$ & True & $0$ & $.03$ & $.37$ & $.75$ & $.91$ & $.97$ & $.99$\\
& Worst & $[0,.16]$ & $[.03,.19]$ & $[.25,.43]$ & $[.51,.76]$ & $[.62,.91]$ &
$[.66,.97]$ & $[.68,.99]$\\
& NSM & $[0,.16]$ & $[.03,.19]$ & $[.31,.43]$ & $[.63,.76]$ & $[.76,.91]$ &
$[.81,.97]$ & $[.83,.99]$\\
& CPQD & $[0,.16]$ & $[.03,.19]$ & $[.25,.43]$ & $[.51,.76]$ & $[.62,.91]$ &
$[.66,.97]$ & $[.68,.99]$\\
& NSM+CPQD & $[0,.16]$ & $[.03,.19]$ & $[.31,.43]$ & $[.63,.76]$ & $[.76,.91]$
& $[.81,.97]$ & $[.83,.99]$\\
& MTR & $[0,.01]$ & $[.04,.12]$ & $[.33,.43]$ & $[.72,.76]$ & $[.90,.91]$ &
$[.96,.97]$ & $[.98,.99]$\\
& NSM+MTR & $[0,.01]$ & $[.04,.12]$ & $[.33,.43]$ & $[.72,.76]$ & $[.90,.91]$ &
$[.96,.97]$ & $[.98,.99]$\\\hline
\end{tabular}
\end{table}

\pagebreak
\newpage

\begin{table}[H]%

\begin{center}
\caption{Identification regions of $F_{0}\left( y \right)$ when $Z \sim Unif\left(z,-z\right)$
}\label{Table2}
\end{center}

\centering$%
\begin{tabular}
[c]{c|c|cccc}\hline
$y$ & $True$ & $z=2$ & $z=1.5$ & $z=1$ & $z=0.5$\\\hline
$-4$ & $0.00$ & $\left[  0,0\right]  $ & $\left[  0,0\right]  $ & $\left[
0,0\right]  $ & $\left[  0,0\right]  $\\
$-2$ & $0.05$ & $\left[  .05,0.06\right]  $ & $\left[  .05,.06\right]  $ &
$\left[  .05,.06\right]  $ & $\left[  .05,.06\right]  $\\
$0$ & $0.50$ & $\left[  .50,.51\right]  $ & $\left[  0.50,0.53\right]  $ &
$\left[  0.48,0.56\right]  $ & $\left[  .45,.59\right]  $\\
$2$ & $0.95$ & $\left[  .94,.95\right]  $ & $\left[  0.92,0.96\right]  $ &
$\left[  0.87,0.97\right]  $ & $\left[  .81,0.98\right]  $\\
$4$ & $1.00$ & $\left[  .99,1.00\right]  $ & $\left[  0.98,1.00\right]  $ &
$\left[  0.96,1.00\right]  $ & $\left[  .93,1.00\right]  $\\
$6$ & $1.00$ & $\left[  1.00,1.00\right]  $ & $\left[  0.99,1.00\right]  $ &
$\left[  0.98,1.00\right]  $ & $\left[  .97,1.00\right]  $\\
$8$ & $1.00$ & $\left[  1.00,1.00\right]  $ & $\left[  1.00,1.00\right]  $ &
$\left[  0.99,1.00\right]  $ & $\left[  .99,1.00\right]  $\\\hline
\end{tabular}
$%
\end{table}%
\begin{table}[H]%

\begin{center}
\caption{Identification regions of $F_{1}\left( y \right)$ when $Z \sim Unif\left(z,-z\right)$
}\label{Table3}
\end{center}

\centering
$%
\begin{tabular}
[c]{c|c|cccc}\hline\hline
$y$ & $True$ & $z=2$ & $z=1.5$ & $z=1$ & $z=0.5$\\\hline
$-4$ & $0.00$ & $\left[  0,0\right]  $ & $\left[  0,0\right]  $ & $\left[
0,0\right]  $ & $\left[  0,0\right]  $\\
$-2$ & $0.01$ & $\left[  .01,.02\right]  $ & $\left[  .01,.03\right]  $ &
$\left[  0,.04\right]  $ & $\left[  .00,.05\right]  $\\
$0$ & $0.18$ & $\left[  .17,.19\right]  $ & $\left[  .16,.21\right]  $ &
$\left[  .14,.25\right]  $ & $\left[  .12,.32\right]  $\\
$2$ & $0.57$ & $\left[  .57,.58\right]  $ & $\left[  .56,.59\right]  $ &
$\left[  .55,.61\right]  $ & $\left[  .53,.66\right]  $\\
$4$ & $0.84$ & $\left[  .84,.84\right]  $ & $\left[  .83,.84\right]  $ &
$\left[  .83,.85\right]  $ & $\left[  .82,.87\right]  $\\
$6$ & $0.94$ & $\left[  .94,.94\right]  $ & $\left[  .94,.94\right]  $ &
$\left[  .94,.95\right]  $ & $\left[  .94,.95\right]  $\\
$8$ & $0.98$ & $\left[  .98,.98\right]  $ & $\left[  .98,.98\right]  $ &
$\left[  .98,.98\right]  $ & $\left[  .98,.98\right]  $\\\hline
\end{tabular}
$%
\end{table}%
\begin{table}[H]%

\begin{center}
\caption{Identification regions of $F_{\Delta }\left( \delta  \right)$ for different values of $z$
}\label{Table4}
\end{center}

\centering
$%
\begin{tabular}
[c]{c|c|cccc}\hline\hline
$\delta$ & $True$ & $z=2$ & $z=1.5$ & $z=1$ & $z=.5$\\\hline
$1$ & $.39$ & $\left[  .01,.78\right]  $ & $\left[  .01,.80\right]  $ &
$\left[  0,.83\right]  $ & $\left[  0,.91\right]  $\\
$3$ & $.78$ & $\left[  .44,.95\right]  $ & $\left[  .38,.95\right]  $ &
$\left[  0.33,.96\right]  $ & $\left[  .25,.97\right]  $\\
$5$ & $.92$ & $\left[  .67,.99\right]  $ & $\left[  .65,.99\right]  $ &
$\left[  0.58,.99\right]  $ & $\left[  .47,.99\right]  $\\
$7$ & $.97$ & $\left[  .84,1.00\right]  $ & $\left[  .80,1.00\right]  $ &
$\left[  .73,1.00\right]  $ & $\left[  .60,1.00\right]  $\\
$9$ & $.99$ & $\left[  .92,1.00\right]  $ & $\left[  .88,1.00\right]  $ &
$\left[  .79,1.00\right]  $ & $\left[  .65,1.00\right]  $\\\hline
\end{tabular}
$%
\end{table}%
\begin{table}[H]%

\begin{center}
\caption{Identification regions of the DTE for different $\rho$
}\label{Table5}
\end{center}

\centering$%
\begin{tabular}
[c]{c|c|ccc}\hline\hline
$\delta$ & $True$ & $\rho=-0.25$ & $\rho=-0.5$ & $\rho=-0.75$\\\hline
$1$ & $0.39$ & $[.01,.83]$ & $[.01,.83]$ & $\left[  0,.83\right]  $\\
$3$ & $0.78$ & $[.38,.95]$ & $[.36,.96]$ & $\left[  .33,.96\right]  $\\
$5$ & $0.92$ & $[.61,.99]$ & $[.60,.99]$ & $\left[  .58,.99\right]  $\\
$7$ & $0.97$ & $[.74,1.00]$ & $[.74,1.00]$ & $\left[  .73,1.00\right]  $\\
$9$ & $0.99$ & $[.80,1.00]$ & $[.80,1.00]$ & $\left[  .79,1.00\right]
$\\\hline
\end{tabular}
$%
\end{table}%
\bigskip
\pagebreak
\bibliography{dissertationbiblio}
\begin{appendix}
\section*{Appendix}

\subsection*{Proof of Lemma \ref{L1.5}}

I\ provide a proof only for sharp bounds on $P_{1}\left(  y,0|z\right)  $.
Sharp bounds on $P_{0}\left(  y,1|z\right)  $ are obtained similarly.
\begin{align*}
&  P\left[  Y_{1}\leq y,0|z\right] \\
&  =P\left[  Y_{1}\leq y,p(z)<U\right] \\
&  =P\left[  Y_{1}\leq y,p(z)<U\leq\overline{p}\right]  +P\left[  Y_{1}\leq
y,\overline{p}<U\right] \\
&  =\underset{p\left(  z\right)  \rightarrow\overline{p}}{\lim}P\left(
y|1,z\right)  \overline{p}-P\left(  y|1,z\right)  p\left(  z\right)  +P\left[
Y_{1}\leq y|\overline{p}<U\right]  \left(  1-\overline{p}\right)  .
\end{align*}
The model (\ref{model}) under $M.1-M.4$ is uninformative about the
counterfactual distribution term $P\left[  Y_{1}\leq y|\overline{p}<U\right]
.$ Therefore by plugging 0 and 1 into the term, \ bounds on $P\left[
Y_{1}\leq y,0|z\right]  $ can be obtained as follows:%

\[
P\left[  Y_{1}\leq y,0|z\right]  \in\left[  L_{10}^{wst}\left(  y,z\right)
,U_{10}^{wst}\left(  y,z\right)  \right]  ,
\]
where
\begin{align*}
L_{10}^{wst}\left(  y,z\right)   &  =\underset{p\left(  z\right)
\rightarrow\overline{p}}{\lim}P\left(  y|1,z\right)  \overline{p}-P\left(
y|1,z\right)  p\left(  z\right)  ,\\
U_{10}^{wst}\left(  y,z\right)   &  =\underset{p\left(  z\right)
\rightarrow\overline{p}}{\lim}P\left(  y|1,z\right)  \overline{p}-P\left(
y|1,z\right)  p\left(  z\right)  +1-\overline{p}.
\end{align*}
${\small \blacksquare}$

\subsection*{Theorem \ref{T1}}

\begin{theorem}
\label{T1}Under $M.1-M.4$, sharp bounds on marginal distributions of $Y_{0}$
and $Y_{1}$, their joint distribution and the DTE are obtained as follows: for
$d\in\left\{  0,1\right\}  $, $y\in\mathbb{R}$, $\delta\in\mathbb{R}$, and
$\left(  y_{0},y_{1}\right)  \in\mathbb{R\times R},$%
\begin{align*}
F_{d}\left(  y\right)   &  \in\left[  F_{d}^{L}\left(  y\right)  ,F_{d}%
^{U}\left(  y\right)  \right]  ,\\
F\left(  y_{0},y_{1}\right)   &  \in\left[  F^{L}\left(  y_{0},y_{1}\right)
,F^{U}\left(  y_{0},y_{1}\right)  \right]  ,\\
F_{\Delta}\left(  \delta\right)   &  \in\left[  F_{\Delta}^{L}\left(
\delta\right)  ,F_{\Delta}^{U}\left(  \delta\right)  \right]  ,
\end{align*}
where
\end{theorem}

\begin{align}
\text{ }F_{0}^{L}\left(  y\right)   &  =\underset{z\in\Xi}{\sup}\left[
P\left\{  y|0,z\right\}  \left(  1-p\left(  z\right)  \right)  +L_{01}%
^{wst}\left(  y,z\right)  \right]  ,\label{worstcase}\\
F_{0}^{U}\left(  y\right)   &  =\underset{z\in\Xi}{\inf}\left[  P\left\{
y|0,z\right\}  \left(  1-p\left(  z\right)  \right)  +U_{01}^{wst}\left(
y,z\right)  \right]  ,\nonumber\\
F_{1}^{L}\left(  y\right)   &  =\underset{z\in\Xi}{\sup}\left[  P\left\{
y|1,z\right\}  p\left(  z\right)  +L_{10}^{wst}\left(  y,z\right)  \right]
,\nonumber\\
F_{1}^{U}\left(  y\right)   &  =\underset{z\in\Xi}{\inf}\left[  P\left\{
y|1,z\right\}  p\left(  z\right)  +U_{10}^{wst}\left(  y,z\right)  \right]
,\nonumber\\
F^{L}\left(  y_{0},y_{1}\right)   &  =\sup_{z\in\Xi}\left[
\begin{array}
[c]{c}%
\max\left\{  \left(  P\left(  y_{0}|0,z\right)  -1\right)  \left(  1-p\left(
z\right)  \right)  +L_{10}^{wst}\left(  y_{1},z\right)  ,0\right\} \\
+\max\left\{  L_{01}^{wst}\left(  y_{0},z\right)  +\left(  P\left(
y_{1}|1,z\right)  -1\right)  p\left(  z\right)  ,0\right\}
\end{array}
\right]  ,\nonumber\\
F^{U}\left(  y_{0},y_{1}\right)   &  =\inf_{z\in\Xi}\left[
\begin{array}
[c]{c}%
\min\left\{  P\left(  y_{0}|0,z\right)  \left(  1-p\left(  z\right)  \right)
,U_{10}^{wst}\left(  y_{1},z\right)  \right\} \\
+\min\left\{  U_{01}^{wst}\left(  y_{0},z\right)  ,P\left(  y_{1}|1,z\right)
p\left(  z\right)  \right\}
\end{array}
\right]  ,\nonumber\\
F_{\Delta}^{L}\left(  \delta\right)   &  =\sup_{z\in\Xi}\left[
\begin{array}
[c]{c}%
\underset{y\in\mathbb{R}}{\sup\max}\left\{  P\left(  y|1,z\right)  p\left(
z\right)  -U_{01}^{wst}\left(  y-\delta,z\right)  ,0\right\} \\
+\underset{y\in\mathbb{R}}{\sup\max}\left\{  L_{10}^{wst}\left(  y,z\right)
-P\left(  y-\delta|0,z\right)  \left(  1-p\left(  z\right)  \right)
,0\right\}
\end{array}
\right]  ,\nonumber\\
F_{\Delta}^{U}\left(  \delta\right)   &  =1+\inf_{z\in\Xi}\left[
\begin{array}
[c]{c}%
\underset{y\in\mathbb{R}}{\inf\min}\left\{  P\left(  y|1,z\right)  p\left(
z\right)  -L_{01}^{wst}\left(  y-\delta,z\right)  ,0\right\} \\
+\underset{y\in\mathbb{R}}{\inf\min}\left\{  U_{10}^{wst}\left(  y,z\right)
-P\left(  y-\delta|0,z\right)  \left(  1-p\left(  z\right)  \right)
,0\right\}
\end{array}
\right]  .\nonumber
\end{align}

\begin{proof}
The proof consists of three parts: sharp bounds on (i) marginal distributions,
(ii) the joint distribution, and (iii) the DTE.\newline\textbf{Part 1. Sharp
bounds on marginal distributions }$F_{0}\left(  \cdot\right)  $\textbf{\ and
}$F_{1}\left(  \cdot\right)  $\newline Since sharp bounds on $F_{0}\left(
y\right)  $ are obtained similarly, I\ derive sharp bounds on $F_{1}\left(
\cdot\right)  $ only. By M.3, $P\left[  Y_{1}\leq y\right]  =P\left[
Y_{1}\leq y|z\right]  $ for any $z\in\Xi$ and $P\left[  Y_{1}\leq y|z\right]
$ can be written as\ the sum of the factual and counterfactual components as
follows:%
\begin{align*}
&  P\left[  Y_{1}\leq y|z\right] \\
&  =P_{1}\left(  y,0|z\right)  +P\left(  y,1|z\right)  .
\end{align*}
Since $P\left[  Y_{1}\leq y,0|z\right]  \in\left[  L_{10}^{wst}\left(
y,z\right)  ,U_{10}^{wst}\left(  y,z\right)  \right]  $ by Lemma \ref{L1.5},%
\begin{align*}
&  P\left(  y|1,z\right)  p\left(  z\right)  +L_{10}^{wst}\left(  y,z\right)
\\
&  \leq P\left[  Y_{1}\leq y|z\right] \\
&  \leq P\left(  y|1,z\right)  p\left(  z\right)  +U_{10}^{wst}\left(
y,z\right)
\end{align*}
Consequently, sharp bounds on $P\left[  Y_{1}\leq y\right]  $ are obtained by
taking the intersection for the bounds on $P\left[  Y_{1}\leq y|z\right]  $
over all $z\in\Xi$ as follows:
\begin{align*}
F_{1}^{L}\left(  y\right)   &  =\sup_{z\in\Xi}\left\{  P\left(  y|1,z\right)
p\left(  z\right)  +L_{10}^{wst}\left(  y,z\right)  \right\}  ,\\
F_{1}^{U}\left(  y\right)   &  =\inf_{z\in\Xi}\left\{  P\left(  y|1,z\right)
p\left(  z\right)  +U_{10}^{wst}\left(  y,z\right)  \right\}  .
\end{align*}
\textbf{Part 2. Sharp bounds on the joint distribution }$F\left(  \cdot
,\cdot\right)  $\newline By M.3,%
\begin{align}
&  F\left(  y_{0},y_{1}\right) \label{A.1}\\
&  =P\left(  Y_{0}\leq y_{0},Y_{1}\leq y_{1}|z\right) \nonumber\\
&  =P\left(  Y_{0}\leq y_{0},Y_{1}\leq y_{1},D=0|z\right)  +P\left(  Y_{0}\leq
y_{0},Y_{1}\leq y_{1},D=1|z\right)  .\nonumber
\end{align}
Note that the model (\ref{model}) and $M.1-M.5$ does not restrict the joint
distribution of $Y_{0}$ and $Y_{1}$ as discussed in Subsection 3.1. Therefore,
for $d\in\left\{  0,1\right\}  ,$ sharp bounds on $P\left(  Y_{0}\leq
y_{0},Y_{1}\leq y_{1}|d,z\right)  $ are obtained by Fr\'{e}chet-Hoeffding
bounds as follows: for any $\left(  y_{0},y_{1}\right)  \in\mathbb{R}^{2},$%
\begin{align*}
&  \max\left\{  P\left(  y_{0}|0,z\right)  +P_{1}\left(  y_{1}|0,z\right)
-1,0\right\} \\
&  \leq P\left(  Y_{0}\leq y_{0},Y_{1}\leq y_{1}|0,z\right) \\
&  \leq\min\left\{  P\left(  y_{0}|0,z\right)  ,P_{1}\left(  y_{1}|0,z\right)
\right\}  .
\end{align*}
Since $P_{1}\left(  y_{1}|0,z\right)  $ is only partially identified, sharp
bounds on $P\left(  Y_{0}\leq y_{0},Y_{1}\leq y_{1}|0,z\right)  $ are obtained
by taking the union over all possible values of $P_{1}\left(  y_{1}%
|0,z\right)  .$ Therefore, sharp bounds on $P\left(  Y_{0}\leq y_{0},Y_{1}\leq
y_{1},D=0|z\right)  =P\left(  Y_{0}\leq y_{0},Y_{1}\leq y_{1}|0,z\right)
\left(  1-p\left(  z\right)  \right)  $ are derived as follows:%
\begin{align*}
&  \max\left\{  P\left(  y_{0},0|z\right)  +L_{10}^{wst}\left(  y,z\right)
-\left(  1-p\left(  z\right)  \right)  ,0\right\} \\
&  \leq P\left(  Y_{0}\leq y_{0},Y_{1}\leq y_{1},D=0|z\right) \\
&  \leq\min\left\{  P\left(  y_{0},0|z\right)  ,U_{10}^{wst}\left(
y,z\right)  \right\}  .
\end{align*}
Similarly,%
\begin{align*}
&  \max\left\{  L_{01}^{wst}\left(  y,z\right)  +\left(  P\left(
y_{1}|1,z\right)  -1\right)  p\left(  z\right)  ,0\right\} \\
&  \leq P\left(  Y_{0}\leq y_{0},Y_{1}\leq y_{1},D=1|z\right) \\
&  \leq\min\left\{  U_{01}^{wst}\left(  y,z\right)  ,P\left(  y_{1}%
|1,z\right)  p\left(  z\right)  \right\}  .
\end{align*}
By (\ref{A.1}), sharp bounds on $P\left(  Y_{0}\leq y_{0},Y_{1}\leq
y_{1}\right)  $ are obtained by taking the intersection of the bounds over all
values of $z\in\Xi,$%
\begin{align*}
F^{L}\left(  y_{0},y_{1}\right)   &  =\sup_{z\in\Xi}\left\{  \max\left\{
\left(  P\left(  y_{0}|0,z\right)  -1\right)  \left(  1-p\left(  z\right)
\right)  +L_{10}^{wst}\left(  y_{1},z\right)  ,0\right\}  \right. \\
&  \left.  +\max\left\{  L_{01}^{wst}\left(  y_{0},z\right)  +\left(  P\left(
y_{1}|1,z\right)  -1\right)  p\left(  z\right)  ,0\right\}  \right\}  ,\\
F^{U}\left(  y_{0},y_{1}\right)   &  =\inf_{z\in\Xi}\left\{  \min\left\{
P\left(  y_{0}|0,z\right)  \left(  1-p\left(  z\right)  \right)  ,U_{10}%
^{wst}\left(  y|z\right)  \right\}  \right. \\
&  \left.  +\min\left\{  U_{01}^{wst}\left(  y_{0},z\right)  ,P\left(
y_{1}|1,z\right)  p\left(  z\right)  \right\}  \right\}  .
\end{align*}
\textbf{Part 3. Sharp bounds on the DTE }$F_{\Delta}\left(  \cdot\right)
$\textbf{\ }\newline As shown in Part 2, the model (\ref{model}) and $M.1-M.4$
do not restrict the joint distribution of $Y_{0}$ and $Y_{1}$ and sharp bounds
on the DTE are obtained by Makarov bounds. Specifically,%
\begin{align*}
&  P\left(  Y_{1}-Y_{0}\leq\delta\right) \\
&  =P\left(  Y_{1}-Y_{0}\leq\delta|z\right) \\
&  =P\left(  Y_{1}-Y_{0}\leq\delta,D=1|z\right)  +P\left(  Y_{1}-Y_{0}%
\leq\delta,D=0|z\right)  .
\end{align*}
Since%
\begin{align*}
P\left(  Y_{1}-Y_{0}\leq\delta,D=0|z\right)   &  =P\left(  Y_{1}-Y_{0}%
\leq\delta|0,z\right)  \left(  1-p\left(  z\right)  \right)  ,\\
P\left(  Y_{1}-Y_{0}\leq\delta,D=1|z\right)   &  =P\left(  Y_{1}-Y_{0}%
\leq\delta|1,z\right)  p\left(  z\right)  ,
\end{align*}
by Makarov bounds,%
\begin{align*}
&  \underset{y\in\mathbb{R}}{\sup}\max\left\{  L_{10}^{wst}\left(  y,z\right)
-P\left(  y-\delta|0,z\right)  \left(  1-p\left(  z\right)  \right)
,0\right\} \\
&  \leq P\left(  Y_{1}-Y_{0}\leq\delta,D=0|z\right) \\
&  \leq\left(  1-p\left(  z\right)  \right)  +\underset{y\in\mathbb{R}}%
{\inf\max}\left\{  U_{10}^{wst}\left(  y|z\right)  -P\left(  y-\delta
|0,z\right)  \left(  1-p\left(  z\right)  \right)  ,0\right\}  ,
\end{align*}
and%
\begin{align*}
&  \underset{y\in\mathbb{R}}{\sup}\max\left\{  P\left(  y|1,z\right)  p\left(
z\right)  -U_{01}^{wst}\left(  y-\delta|z\right)  ,0\right\} \\
&  \leq P\left(  Y_{1}-Y_{0}\leq\delta,D=1|z\right) \\
&  \leq p\left(  z\right)  +\underset{y\in\mathbb{R}}{\inf\max}\left\{
P\left(  y|1,z\right)  p\left(  z\right)  -L_{01}^{wst}\left(  y-\delta
|z\right)  ,0\right\}  .
\end{align*}
Therefore, sharp bounds on the DTE are obtained from the intersection bounds
as follows:%
\begin{align*}
&  \sup_{z\in\Xi}\left\{  \underset{y\in\mathbb{R}}{\sup\max}\left\{
L_{10}^{wst}\left(  y,z\right)  -P\left(  y-\delta|0,z\right)  \left(
1-p\left(  z\right)  \right)  ,0\right\}  \right. \\
&  \left.  +\underset{y\in\mathbb{R}}{\sup\max}\left\{  P\left(  y|1,z\right)
p\left(  z\right)  -U_{01}^{wst}\left(  y-\delta|z\right)  ,0\right\}
\right\} \\
&  \leq P\left(  Y_{1}-Y_{0}\leq\delta\right) \\
&  \leq1+\inf_{z\in\Xi}\left\{  \underset{y\in\mathbb{R}}{\inf\max}\left\{
P\left(  y|1,z\right)  p\left(  z\right)  -L_{01}^{wst}\left(  y-\delta
,z\right)  ,0\right\}  \right. \\
&  \left.  +\underset{y\in\mathbb{R}}{\inf\max}\left\{  U_{10}^{wst}\left(
y,z\right)  -P\left(  y-\delta|0,z\right)  \left(  1-p\left(  z\right)
\right)  ,0\right\}  \right\}  .
\end{align*}

\end{proof}

\subsection*{Corollary \ref{T2}}

\begin{corollary}
\label{T2}(Bounds on the marginal distributions of potential outcomes) Under
$M.1-M.4$ and SM, sharp bounds on marginal distributions of $Y_{0}$ and
$Y_{1}$, their joint distribution and the DTE are given as follows: for
$d\in\left\{  0,1\right\}  $, $y\in\mathbb{R}$, $\delta\in\mathbb{R}$, and
$\left(  y_{0},y_{1}\right)  \in\mathbb{R\times R},$%
\begin{align*}
F_{d}\left(  y\right)   &  \in\left[  F_{d}^{L}\left(  y\right)  ,F_{d}%
^{U}\left(  y\right)  \right]  ,\\
F\left(  y_{0},y_{1}\right)   &  \in\left[  F^{L}\left(  y_{0},y_{1}\right)
,F^{U}\left(  y_{0},y_{1}\right)  \right]  ,\\
F_{\Delta}\left(  \delta\right)   &  \in\left[  F_{\Delta}^{L}\left(
\delta\right)  ,F_{\Delta}^{U}\left(  \delta\right)  \right]  ,
\end{align*}
where%
\begin{align*}
\text{ }F_{0}^{L}\left(  y\right)   &  =\underset{z\in\Xi}{\sup}\left[
P\left(  y|0,z\right)  \left(  1-p\left(  z\right)  \right)  +L_{01}%
^{wst}\left(  y,z\right)  \right]  ,\\
F_{0}^{U}\left(  y\right)   &  =\underset{z\in\Xi}{\inf}\left[  P\left(
y|0,z\right)  \left(  1-p\left(  z\right)  \right)  +U_{01}^{sm}\left(
y,z\right)  \right]  ,\\
F_{1}^{L}\left(  y\right)   &  =\underset{z\in\Xi}{\sup}\left[  P\left(
y|1,z\right)  p\left(  z\right)  +L_{10}^{wst}\left(  y,z\right)  \right]  ,\\
F_{1}^{U}\left(  y\right)   &  =\underset{z\in\Xi}{\inf}\left[  P\left(
y|1,z\right)  p\left(  z\right)  +U_{10}^{sm}\left(  y,z\right)  \right]  ,\\
F^{L}\left(  y_{0},y_{1}\right)   &  =\sup_{z\in\Xi}\left[
\begin{array}
[c]{c}%
\max\left\{  \left(  P\left(  y_{0}|0,z\right)  -1\right)  \left(  1-p\left(
z\right)  \right)  +L_{10}^{wst}\left(  y_{1},z\right)  ,0\right\} \\
+\max\left\{  L_{01}^{wst}\left(  y_{0},z\right)  +\left(  P\left(
y_{1}|1,z\right)  -1\right)  p\left(  z\right)  ,0\right\}
\end{array}
\right]  ,\\
F^{U}\left(  y_{0},y_{1}\right)   &  =\inf_{z\in\Xi}\left[
\begin{array}
[c]{c}%
\min\left\{  P\left(  y_{0}|0,z\right)  \left(  1-p\left(  z\right)  \right)
,U_{10}^{sm}\left(  y_{1},z\right)  \right\} \\
+\min\left\{  U_{01}^{sm}\left(  y_{0},z\right)  ,P\left(  y_{1}|1,z\right)
p\left(  z\right)  \right\}
\end{array}
\right]  ,\\
F_{\Delta}^{L}\left(  \delta\right)   &  =\sup_{z\in\Xi}\left[
\begin{array}
[c]{c}%
\underset{y\in\mathbb{R}}{\sup\max}\left\{  P\left(  y|1,z\right)  p\left(
z\right)  -U_{01}^{sm}\left(  y-\delta,z\right)  ,0\right\} \\
+\underset{y\in\mathbb{R}}{\sup\max}\left\{  L_{10}^{wst}\left(  y,z\right)
-P\left(  y-\delta|0,z\right)  \left(  1-p\left(  z\right)  \right)
,0\right\}
\end{array}
\right]  ,\\
F_{\Delta}^{U}\left(  \delta\right)   &  =1+\inf_{z\in\Xi}\left[
\begin{array}
[c]{c}%
\underset{y\in\mathbb{R}}{\inf\min}\left\{  P\left(  y|1,z\right)  p\left(
z\right)  -L_{01}^{wst}\left(  y-\delta,z\right)  ,0\right\} \\
+\underset{y\in\mathbb{R}}{\inf\min}\left\{  U_{10}^{sm}\left(  y,z\right)
-P\left(  y-\delta|0,z\right)  \left(  1-p\left(  z\right)  \right)
,0\right\}
\end{array}
\right]  .
\end{align*}

\end{corollary}

\subsection*{Theorem \ref{T3}}

\ \ \ \ \ \ \ \ \ \ \ \ \ \ \ \ \ \ \ \ \ \ \ \ \ \ \ \ \ \ \ \ \ \ \ \ \ \ \ \ \ \ \ \ \ \ \ \ \ \ \ \ \ \ \ \ \ \ \ \ \ \ \ \ \ \ \ \ \ \ \ \ \ \ \ \ \ \ \ \ \ \ \ \ \ \ \ \ \ \ \ \ \ \ \ \ \ \ \ \ \ \ \ \ \ \ \ \ \ \ \ \ \ \ \ \ \ \ \ \ \ \ \ \ \ \ \ \ \ \ \ \ \ \ \ \ \ \ \ \ \ \ \ \ \ \ \ \ \ \ \ \ \ \ \ \ \ \ \ \ \ \ \ \ \ \ \ \ \ \ \ \ \ \ \ \ \ \ \ \ \ \ \ \ \ \ \ \ \ \ \ \ \ \ \ \ \ \ \ \ \ \ \ \ \ \ \ \ \ \ \ \ \ \ \ \ \ \ \ \ \ \ \ \ \ \ \ \ \ \ \ \ \ \ \ \ \ \ \ \ \ \ \ \ \ \ \ \ \ \ \ \ \ \ \ \ \ \ \ \ \ \ \ \ \ \ \ \ \ \ \ \ \ \ \ \ \ \ \ \ \ \ \ \ \ \ \ \ \ \ \ \ \ \ \ \ \ \ \ \ \ \ \ \ \ \ \ \ \ \ \ \ \ \ \ \ \ \ \ \ \ \ \ \ \ \ \ \ \ \ \ \ \ \ \ \ \ \ \ \ \ \ \ \ \ \ \ \ \ \ \ \ \ \ 

\begin{theorem}
\label{T3}Under $M.1-M.5$, and CPQD, sharp bounds on $F_{0}\left(
y_{0}\right)  $, $F_{1}\left(  y_{1}\right)  ,$ and $F_{\Delta}\left(
\delta\right)  $ are identical to those given in Theorem \ref{T1}. Sharp
bounds on $F\left(  y_{0},y_{1}\right)  $ are obtained as follows: for
$\left(  y_{0},y_{1}\right)  \in\mathbb{R\times R},$%
\[
F\left(  y_{0},y_{1}\right)  \in\left[  F^{L}\left(  y_{0},y_{1}\right)
,F^{U}\left(  y_{0},y_{1}\right)  \right]  ,
\]
where%
\begin{align*}
F_{d}\left(  y\right)   &  \in\left[  F_{d}^{L}\left(  y\right)  ,F_{d}%
^{U}\left(  y\right)  \right]  ,\\
F\left(  y_{0},y_{1}\right)   &  \in\left[  F^{L}\left(  y_{0},y_{1}\right)
,F^{U}\left(  y_{0},y_{1}\right)  \right]  ,\\
F_{\Delta}\left(  \delta\right)   &  \in\left[  F_{\Delta}^{L}\left(
\delta\right)  ,F_{\Delta}^{U}\left(  \delta\right)  \right]  ,
\end{align*}%
\begin{align*}
F^{L}\left(  y_{0},y_{1}\right)   &  =\underset{z\in\Xi}{\sup}\left\{
P\left(  y_{0}|0,z\right)  L_{10}^{wst}\left(  y_{1},z\right)  +L_{01}%
^{wst}\left(  y_{0},z\right)  P\left(  y_{1}|1,z\right)  \right\}  ,\\
F^{U}\left(  y_{0},y_{1}\right)   &  =\inf_{z\in\Xi}\left[
\begin{array}
[c]{c}%
\min\left\{  P\left(  y_{0}|0,z\right)  \left(  1-p\left(  z\right)  \right)
,U_{10}^{wst}\left(  y,z\right)  \right\} \\
+\min\left\{  U_{01}^{wst}\left(  y_{0},z\right)  ,P\left(  y_{1}|1,z\right)
p\left(  z\right)  \right\}
\end{array}
\right]  .
\end{align*}

\end{theorem}

\begin{proof}
The proof of Theorem \ref{T3} consists of two parts: sharp bounds on the joint
distribution of $Y_{0}$ and $Y_{1}$ and sharp bounds on the DTE under
$M.1-M.5$ and CPQD.\newline\textbf{Part 1. Sharp bounds on the joint
distribution of }$Y_{0}$\textbf{\ and }$Y_{1}$\newline In Subsection
$\backslash$%
ref\{c2s3s3\}, I\ proved that
\begin{align*}
P\left(  Y\leq y_{0},Y_{1}\leq y_{1}|0,z\right)   &  \geq\left(  \frac
{1}{1-p\left(  z\right)  }\right)  ^{2}\left(  Y_{0}\leq y_{0}|0,z\right)
P\left(  Y_{1}\leq y_{1}|0,z\right)  ,\\
P\left(  Y_{0}\leq y_{0},Y\leq y_{1}|1,z\right)   &  \geq\left(  \frac
{1}{p\left(  z\right)  }\right)  ^{2}P\left(  Y_{0}\leq y_{0}|1,z\right)
P\left(  Y_{1}\leq y_{1}|1,z\right)  .
\end{align*}
Also by (\ref{*1}) and (\ref{*2}), for any $z\in\Xi,$%
\begin{align*}
&  P\left(  Y_{0}\leq y_{0},Y_{1}\leq y_{1}\right) \\
&  =P\left(  Y_{0}\leq y_{0},Y_{1}\leq y_{1}|z\right) \\
&  =P\left(  Y\leq y_{0},Y_{1}\leq y_{1}|0,z\right)  \left(  1-p\left(
z\right)  \right)  +P\left(  Y_{0}\leq y_{0},Y\leq y_{1}|1,z\right)  p\left(
z\right) \\
&  \geq P\left(  y_{0}|0,z\right)  L_{10}^{wst}\left(  y_{1},z\right)
+L_{01}^{wst}\left(  y_{1},z\right)  P\left(  y_{1}|1,z\right)
\end{align*}
Finally, the lower bound $P\left(  Y_{0}\leq y_{0},Y_{1}\leq y_{1}\right)  $
can be obtained by taking the intersection over all $z\in\Xi,$
\begin{align*}
&  P\left(  Y_{0}\leq y_{0},Y_{1}\leq y_{1}\right) \\
&  \geq\underset{z\in\Xi}{\sup}\left\{  P\left(  y_{0}|0,z\right)
L_{10}^{wst}\left(  y_{1},z\right)  +L_{01}^{wst}\left(  y_{1},z\right)
P\left(  y_{1}|1,z\right)  \right\}  .
\end{align*}
The upper bound is obtained as Fr\'{e}chet-Hoeffing upper bound as follows:%
\begin{align*}
&  P\left(  Y_{0}\leq y_{0},Y_{1}\leq y_{1}\right) \\
&  \leq\underset{z\in\Xi}{\inf}\left\{  \min\left\{  P\left(  y_{0}%
|0,z\right)  ,P_{1}\left(  y_{1}|0,z\right)  \right\}  \left(  1-p\left(
z\right)  \right)  \right. \\
&  \left.  +\min\left\{  P_{0}\left(  y_{0}|1,z\right)  ,P\left(
y_{1}|1,z\right)  \right\}  p\left(  z\right)  \right\}  .
\end{align*}
The lower bound is obtained when $\varepsilon_{0}$ and $\varepsilon_{1}$ are
independent conditionally on $U$, while the upper bound is obtained when
$\varepsilon_{0}$ and $\varepsilon_{1}$ are perfectly dependent conditionally
on $U$. Thus they are sharp.\newline\textbf{Part 2. Sharp bounds on the
DTE}\newline To show that CPQD has no additional identifying power on the
DTE, I\ use the following Lemma which has been presented by \citet{WD1990} and \citet{FP2009}.\newline\newline\textbf{Lemma B.1 }Let $\underline{C}$ denote
a lower bound on the copula of $X$ and $Y$, and $F_{X+Y}$ denote the
distribution function of $X+Y.$\ If support of $(X,Y),$ $supp(X,Y)$ satisfies
$supp(X,Y)=supp(X)\times supp(Y),$%
\[
\sup_{x+y=z}\underline{C}\left(  F_{X}\left(  x\right)  ,F_{Y}\left(
y\right)  \right)  \leq F_{X+Y}\left(  z\right)  \leq\inf_{x+y=z}%
\underline{C^{d}}\left(  F_{X}\left(  x\right)  ,F_{Y}\left(  y\right)
\right)
\]
where $\underline{C^{d}}\left(  u,v\right)  =u+v-\underline{C}\left(
u,v\right)  .$

Let $Y_{1}=X$ and $Y_{0}=-Y$. By Lemma B.1, \ sharp bounds on the DTE are
affected by only the upper bound on the copula$\ $of $Y_{0}$ and $Y_{1}.$
Since CPQD improves only the lower bound on the copula if $Y_{0}$ and $Y_{1}$,
the DTE\ bounds do not improve by CPQD. \newline
\end{proof}

\subsection*{Theorem \ref{T4}}

\begin{theorem}
\label{T4}Under $M.1-M.4$ and MTR, sharp bounds on $F\left(  y_{0}%
,y_{1}\right)  $, and $F_{\Delta}\left(  \delta\right)  $ are given as
follows: for $d\in\left\{  0,1\right\}  $, $y\in\mathbb{R}$, $\delta
\in\mathbb{R}$, and $\left(  y_{0},y_{1}\right)  \in\mathbb{R\times R},$%
\begin{align*}
F\left(  y_{0},y_{1}\right)   &  \in\left[  F^{L}\left(  y_{0},y_{1}\right)
,F^{U}\left(  y_{0},y_{1}\right)  \right]  ,\\
F_{\Delta}\left(  \delta\right)   &  \in\left[  F_{\Delta}^{L}\left(
\delta\right)  ,F_{\Delta}^{U}\left(  \delta\right)  \right]  ,
\end{align*}
where%
\begin{align*}
&  F^{L}\left(  y_{0},y_{1}\right) \\
&  =\left\{
\begin{array}
[c]{cc}%
\begin{array}
[c]{c}%
\underset{z\in\Xi}{\sup}\left[  \max\left\{  \underset{y_{0}\leq y\leq y_{1}%
}{\sup}\left\{
\begin{array}
[c]{c}%
\left(  P\left(  y_{0}|0,z\right)  -P\left(  y|0,z\right)  \right)  \left(
1-p\left(  z\right)  \right) \\
+L_{10}^{wst}\left(  y,z\right)
\end{array}
\right\}  ,0\right\}  \right. \\
\left.  +\max\left\{  \underset{y_{0}\leq y\leq y_{1}}{\sup}\left\{
L_{01}^{mtr}\left(  y_{0},z\right)  -U_{01}^{wst}\left(  y,z\right)  +\left(
P\left(  Y\leq y|1,z\right)  \right)  p\left(  z\right)  \right\}  ,0\right\}
\right]  ,
\end{array}
& \text{if }y_{0}<y_{1},\\
F_{1}^{L}\left(  y\right)  , & \text{if }y_{0}\geq y_{1},
\end{array}
\right.
\end{align*}%
\begin{align*}
F^{U}\left(  y_{0},y_{1}\right)   &  =\left\{
\begin{array}
[c]{cc}%
\begin{array}
[c]{c}%
\underset{z\in\Xi}{\inf}\left\{  \min\left\{  P\left(  Y\leq y_{0}|0,z\right)
\left(  1-p\left(  z\right)  \right)  ,U_{10}^{mtr}\left(  y,z\right)
\right\}  \right. \\
\left.  +\min\left\{  U_{01}^{wst}\left(  y,z\right)  ,P\left(  y_{1}%
|1,z\right)  p\left(  z\right)  \right\}  \right\}  ,
\end{array}
& \text{if }y_{0}<y_{1},\\
F_{1}^{U}\left(  y\right)  , & \text{if }y_{0}\geq y_{1},
\end{array}
\right. \\
F_{\Delta}^{U}\left(  \delta\right)   &  =1+\inf_{z\in\Xi}\left\{  p\left(
z\right)  +\underset{y\in\mathbb{R}}{\inf}\max\left\{  P\left(  y|1,z\right)
p\left(  z\right)  -L_{01}^{mtr}\left(  y-\delta,z\right)  ,0\right\}  \right.
\\
&  \left.  +\underset{y\in\mathbb{R}}{\inf}\max\left\{  U_{10}^{mtr}\left(
y,z\right)  -P\left(  y-\delta|0,z\right)  \left(  1-p\left(  z\right)
\right)  ,0\right\}  \right\}  ,\\
F_{\Delta}^{L}\left(  \delta\right)   &  =\sup_{z\in\Xi}\left\{
\underset{\left\{  a_{k}\right\}  _{k=-\infty}^{\infty}\in\mathcal{A}_{\delta
}}{\sup\max}\left\{  P\left(  a_{k+1}|1,z\right)  p\left(  z\right)
-U_{01}^{wst}\left(  a_{k},z\right)  ,0\right\}  \right. \\
&  \left.  +\underset{\left\{  b_{k}\right\}  _{k=-\infty}^{\infty}%
\in\mathcal{A}_{\delta}}{\sup\max}\left\{  L_{10}^{wst}\left(  b_{k+1}%
,z\right)  -P\left(  b_{k}|0,z\right)  \left(  1-p\left(  z\right)  \right)
,0\right\}  \right\}  ,
\end{align*}
where%
\[
\mathcal{A}_{\delta}=\left\{  \left\{  a_{k}\right\}  _{k=-\infty}^{\infty
};0\leq a_{k+1}-a_{k}\leq\delta\text{ for every integer }k\right\}  .
\]

\end{theorem}

\begin{proof}
The proof of Theorem \ref{T4} considers sharp bounds on the joint distribution
of $Y_{0}$ and $Y_{1}$ only. Sharp bounds on the marginal distributions have
been derived in Subsection \ref{c2s3s4} and sharp bounds on the DTE are
trivially derived from Lemma \ref{L5}.\newline\textbf{Part 1. Sharp bounds on
the joint distribution of }$Y_{0}$\textbf{\ and }$Y_{1}$\newline Under MTR, it
is obvious that $F\left(  y_{0},y_{1}\right)  =$ $F_{1}\left(  y_{1}\right)  $
for $y_{1}\leq y_{0}.$ Throughout this proof, I\ consider only the nontrivial
case $y_{0}<y_{1}.$\newline To obtain sharp bounds on the joint distribution
under $M.1-M.5$ and MTR, I\ use the following Lemma B.2 presented by
\citet{N2006}.\newline\newline\textbf{Lemma B.2} Let $C$ be a copula, and
suppose $C\left(  a,b\right)  =\theta,$ where $\left(  a,b\right)  $ is in
$\left(  0,1\right)  ^{2}$ and $\theta$\ satisfies $\max\left(
a+b-1,0\right)  \leq\theta\leq\min\left(  a,b\right)  $. Then%
\[
C_{L}\left(  u,v\right)  \leq C\left(  u,v\right)  \leq C_{U}\left(
u,v\right)  ,
\]
where $C_{U}$\ and $C_{L}$\ are the copulas given by
\begin{align*}
C_{U}\left(  u,v\right)   &  =\min\left(  u,v,\theta+\left(  u-a\right)
^{+}+\left(  v-b\right)  ^{+}\right)  ,\\
C_{L}\left(  u,v\right)   &  =\max\left(  0,u+v-1,\theta-\left(  a-u\right)
^{+}-\left(  b-v\right)  ^{+}\right)  .
\end{align*}
where $\left(  x\right)  ^{+}=\max\left\{  x,0\right\}  $.\newline%
\newline\textbf{Lemma B.3 }For fixed marginal distribution functions $F_{0}$
and $F_{1},$\ sharp bounds on the joint distribution function $F$ are given as
follows:%
\[
F^{L}\left(  y_{0},y_{1}\right)  \leq F\left(  y_{0},y_{1}\right)  \leq
F^{U}\left(  y_{0},y_{1}\right)
\]
where%
\begin{align*}
F^{L}\left(  y_{0},y_{1}\right)   &  =\max_{y_{0}\leq y<y_{1}}\left\{
F_{1}\left(  y\right)  -F_{0}\left(  y\right)  +F_{0}\left(  y_{0}\right)
\right\}  ,\\
F^{U}\left(  y_{0},y_{1}\right)   &  =\inf_{y\in\mathbb{R}}\min\left(
F_{0}\left(  y_{0}\right)  ,F_{1}\left(  y_{1}\right)  \right)  .
\end{align*}
\newline\newline From Lemma B.3, sharp bounds on the joint distribution are
readily obtained as follows: if $y_{0}<y_{1},$%
\begin{align*}
F^{L}\left(  y_{0},y_{1}\right)   &  =\sup_{z\in\Xi}\left[  \max\left\{
\underset{y_{0}\leq y\leq y_{1}}{\sup}\left\{  \left(  P\left(  y_{0}%
|0,z\right)  -P\left(  y|0,z\right)  \right)  \left(  1-p\left(  z\right)
\right)  +L_{10}^{wst}\left(  y,z\right)  \right\}  ,0\right\}  \right. \\
&  \left.  +\max\left\{  \underset{y_{0}\leq y\leq y_{1}}{\sup}\left\{
L_{01}^{mtr}\left(  y_{0}|z\right)  -U_{01}^{wst}\left(  y,z\right)  +\left(
P\left(  y|1,z\right)  \right)  p\left(  z\right)  \right\}  ,0\right\}
\right]  ,\\
F^{U}\left(  y_{0},y_{1}\right)   &  =\inf_{z\in\Xi}\left\{  \min\left\{
P\left(  y_{0}|0,z\right)  \left(  1-p\left(  z\right)  \right)  ,U_{10}%
^{mtr}\left(  y,z\right)  \right\}  \right. \\
&  \left.  +\min\left\{  U_{01}^{wst}\left(  y,z\right)  ,P\left(
y_{1}|1,z\right)  p\left(  z\right)  \right\}  \right\}  .
\end{align*}

\end{proof}

\begin{proof}
[Proof of Lemma B.3]Since MTR is equivalent to the condition that $F\left(
y,y\right)  =F_{1}\left(  y\right)  $ for any $y\in\mathbb{R}$, by Lemma B.2
the lower and upper bounds on $F\left(  y_{0},y_{1}\right)  $ are obtained by
taking the intersection over all $y\in\mathbb{R}$ as follows:
\begin{align*}
F^{U}\left(  y_{0},y_{1}\right)   &  =\inf_{y\in\mathbb{R}}\min\left(
F_{0}\left(  y_{0}\right)  ,F_{1}\left(  y_{1}\right)  ,F_{1}\left(  y\right)
+\left(  F_{0}\left(  y_{0}\right)  -F_{0}\left(  y\right)  \right)
^{+}+\left(  F_{1}\left(  y_{1}\right)  -F_{1}\left(  y\right)  \right)
^{+}\right)  ,\\
F^{L}\left(  y_{0},y_{1}\right)   &  =\sup_{y\in\mathbb{R}}\max\left(
0,F_{0}\left(  y_{0}\right)  +F_{1}\left(  y_{1}\right)  -1,F_{1}\left(
y\right)  -\left(  F_{0}\left(  y\right)  -F_{0}\left(  y_{0}\right)  \right)
^{+}-\left(  F_{1}\left(  y\right)  -F_{1}\left(  y_{1}\right)  \right)
^{+}\right)  .
\end{align*}
Note that
\begin{align*}
&  \inf_{y\in\mathbb{R}}\left\{  F_{1}\left(  y\right)  +\left(  F_{0}\left(
y_{0}\right)  -F_{0}\left(  y\right)  \right)  ^{+}+\left(  F_{1}\left(
y_{1}\right)  -F_{1}\left(  y\right)  \right)  ^{+}\right\} \\
&  \geq\inf_{y\in\mathbb{R}}\left\{  F_{1}\left(  y\right)  +\left(
F_{1}\left(  y_{1}\right)  -F_{1}\left(  y\right)  \right)  ^{+}\right\} \\
&  \geq\inf_{y\in\mathbb{R}}\left\{  F_{1}\left(  y\right)  +F_{1}\left(
y_{1}\right)  -F_{1}\left(  y\right)  \right\}  =F_{1}\left(  y_{1}\right)  .
\end{align*}
Therefore,%
\[
F^{U}\left(  y_{0},y_{1}\right)  =\min\left(  F_{0}\left(  y_{0}\right)
,F_{1}\left(  y_{1}\right)  \right)  .
\]
Now to derive the lower bound $F^{L}\left(  y_{0},y_{1}\right)  ,$ let
$G\left(  y\right)  =$ $F_{1}\left(  y\right)  -\left(  F_{0}\left(  y\right)
-F_{0}\left(  y_{0}\right)  \right)  ^{+}-\left(  F_{1}\left(  y\right)
-F_{1}\left(  y_{1}\right)  \right)  ^{+}.$ Then for $y_{0}<y_{1,}$%
\[
G\left(  y\right)  =\left\{
\begin{array}
[c]{cc}%
F_{0}\left(  y_{0}\right)  +F_{1}\left(  y_{1}\right)  -F_{0}\left(  y\right)
, & \text{if }y_{1}\leq y\\
F_{1}\left(  y\right)  -F_{0}\left(  y\right)  +F_{0}\left(  y_{0}\right)  , &
\text{if }y_{0}\leq y<y_{1}\\
F_{1}\left(  y\right)  , & \text{if }y<y_{0}%
\end{array}
\right.  .
\]
and so,
\[
\sup_{y\in\mathbb{R}}G\left(  y\right)  =\sup_{y_{0}\leq y\leq y_{1}}\left\{
F_{1}\left(  y\right)  -F_{0}\left(  y\right)  +F_{0}\left(  y_{0}\right)
\right\}
\]
Since $F_{1}\left(  y_{1}\right)  -F_{0}\left(  y_{1}\right)  +F_{0}\left(
y_{0}\right)  \geq\max\left(  0,F_{0}\left(  y_{0}\right)  +F_{1}\left(
y_{1}\right)  -1\right)  ,$ for $y_{0}<y_{1},$
\[
F^{L}\left(  y_{0},y_{1}\right)  =\sup_{y_{0}\leq y\leq y_{1}}\left\{
F_{1}\left(  y\right)  -F_{0}\left(  y\right)  +F_{0}\left(  y_{0}\right)
\right\}  .
\]

\end{proof}

\subsection*{Corollary \ref{T5}}

\begin{corollary}
\label{T5}(Bounds on the marginal distributions of potential outcomes) Under
$M.1-M.4$, PSM and MTR, sharp bounds on marginal distributions of $Y_{0}$ and
$Y_{1}$, their joint distribution and the DTE are given as follows:%
\begin{align*}
F_{0}^{L}\left(  y\right)   &  =\underset{z\in\Xi}{\sup}\left[  P\left(
y|0,z\right)  \left(  1-p\left(  z\right)  \right)  +L_{01}^{mtr}\left(
y,z\right)  \right]  ,\\
F_{0}^{U}\left(  y\right)   &  =\underset{z\in\Xi}{\inf}\left[  P\left(
y|0,z\right)  \left(  1-p\left(  z\right)  \right)  +U_{01}^{sm}\left(
y,z\right)  \right]  ,\\
F_{1}^{L}\left(  y\right)   &  =\underset{z\in\Xi}{\sup}\left[  P\left(
y|1,z\right)  p\left(  z\right)  +L_{10}^{sm}\left(  y,z\right)  \right]  ,\\
F_{1}^{U}\left(  y\right)   &  =\underset{z\in\Xi}{\inf}\left[  P\left(
y|1,z\right)  p\left(  z\right)  +U_{10}^{mtr}\left(  y,z\right)  \right]  ,\\
F^{L}\left(  y_{0},y_{1}\right)   &  =\underset{z\in\Xi}{\sup}\left\{
P\left(  y_{0}|0,z\right)  L_{10}^{sm}\left(  y_{1},z\right)  +L_{01}%
^{mtr}\left(  y_{0},z\right)  P\left(  y_{1}|1,z\right)  \right\}  ,\\
F^{U}\left(  y_{0},y_{1}\right)   &  =\inf_{z\in\Xi}\left[
\begin{array}
[c]{c}%
\min\left\{  P\left(  y_{0}|0,z\right)  \left(  1-p\left(  z\right)  \right)
,U_{10}^{mtr}\left(  y,z\right)  \right\} \\
+\min\left\{  U_{01}^{sm}\left(  y_{0},z\right)  ,P\left(  y_{1}|1,z\right)
p\left(  z\right)  \right\}
\end{array}
\right]  ,\\
F_{\Delta}^{U}\left(  \delta\right)   &  =1+\inf_{z\in\Xi}\left\{  p\left(
z\right)  +\underset{y\in\mathbb{R}}{\inf}\max\left\{  P\left(  y|1,z\right)
p\left(  z\right)  -L_{01}^{mtr}\left(  y-\delta,z\right)  ,0\right\}  \right.
\\
&  \left.  +\underset{y\in\mathbb{R}}{\inf}\max\left\{  U_{10}^{mtr}\left(
y,z\right)  -P\left(  y-\delta|0,z\right)  \left(  1-p\left(  z\right)
\right)  ,0\right\}  \right\}  ,\\
F_{\Delta}^{L}\left(  \delta\right)   &  =\sup_{z\in\Xi}\left\{
\underset{\left\{  a_{k}\right\}  _{k=-\infty}^{\infty}\in\mathcal{A}_{\delta
}}{\sup\max}\left\{  P\left(  a_{k+1}|1,z\right)  p\left(  z\right)
-U_{01}^{wst}\left(  a_{k},z\right)  ,0\right\}  \right. \\
&  \left.  +\underset{\left\{  b_{k}\right\}  _{k=-\infty}^{\infty}%
\in\mathcal{A}_{\delta}}{\sup\max}\left\{  L_{10}^{wst}\left(  b_{k+1}%
,z\right)  -P\left(  b_{k}|0,z\right)  \left(  1-p\left(  z\right)  \right)
,0\right\}  \right\}  ,\\
\text{where }\mathcal{A}_{\delta}  &  =\left\{  \left\{  a_{k}\right\}
_{k=-\infty}^{\infty};0\leq a_{k+1}-a_{k}\leq\delta\text{ for every integer
}k\right\}
\end{align*}

\end{corollary}

\end{appendix}

\end{document}